\documentclass[12pt,reqno]{amsart}
\usepackage{graphicx}
\usepackage{amscd}
\usepackage{amsmath}
\usepackage{epsfig}
\usepackage{amsfonts}
\usepackage{amssymb}

\providecommand{\U}[1]{\protect\rule{.1in}{.1in}}
\providecommand{\U}[1]{\protect\rule{.1in}{.1in}}
\allowdisplaybreaks
\newtheorem{theorem}{Theorem}
\theoremstyle{plain}

\newtheorem{lemma}{Lemma}

\numberwithin{equation}{section}

\input{tcilatex}

\begin{document}
\title[Wave Packets and Coherent Structures]{Wave Packets and Coherent
Structures\\
for Nonlinear Schr\"{o}dinger equations \\
in Variable Nonuniform Media}
\author{Alex Mahalov}
\address{School of Mathematical and Statistical Sciences \& Mathematical,
Computational and Modeling Sciences Center, Arizona State University, Tempe,
AZ 85287--1804, U.S.A.}
\email{mahalov@asu.edu}
\author{Sergei K. Suslov}
\address{School of Mathematical and Statistical Sciences \& Mathematical,
Computational and Modeling Sciences Center, Arizona State University, Tempe,
AZ 85287--1904, U.S.A.}
\email{sks@asu.edu}
\urladdr{http://hahn.la.asu.edu/\symbol{126}suslov/index.html}
\date{\today }

\begin{abstract}
We determine conditions under which a generic gauge invariant nonautonomous
and inhomogeneous nonlinear partial differential equation in the
two-dimensional space-time continuum can be transform into standard
autonomous forms. In addition to the nonlinear Schr\"{o}dinger equation,
important examples include the derivative nonlinear Schr\"{o}dinger
equation, the quintic complex Ginzburg--Landau equation, and the
Gerdjikov--Ivanov equation. This approach provides a mathematical
description of nonstationary media supporting unidimensional signal
propagation and/or total field trapping. In particular, we study
self-similar and nonspreading wave packets for Schr\"{o}dinger equations.
Some other important coherent structures are also analyzed and applications
to nonlinear waves in inhomogeneous media such as atmospheric plasma, fiber
optics, hydrodynamics, and Bose--Einstein condensation are discussed.
\end{abstract}

\maketitle

\section{Introduction}

The nonlinear Schr\"{o}dinger equation is a universal mathematical model for
many physical systems --- in a variety of nonlinear problems, the
one-dimensional nonlinear Schr\"{o}dinger equation with a cubic nonlinearity
describes the propagation of an envelope wave riding over a carrier when the
lowest order contribution to the dispersion relation is proportional to the
square of a local amplitude of the wave envelope. This equation is generic
arising in consideration of the lowest effects of dispersion and
nonlinearity on a wave packet, including nonlinear optics, deep water waves,
acoustics, plasma physics and nonlinear propagation of plane diffracted wave
beams in the focusing regions of the ionosphere \cite{Balakrish85}, \cite%
{BenneyNewell67}, \cite{Bouchbinder03}, \cite{ChenHH:LiuCS76}, \cite%
{ChenHH:LiuCS78}, \cite{GurevichNLIonosphere}, \cite{KadomtsevCollectPP}, 
\cite{KadKarp71}, \cite{KarpmanNW}, \cite{Kelley65}, \cite{Scott03}, \cite%
{ScottEncNS}, \cite{OstrPot99}, \cite{ZakhOstr09}, and \cite{Zakh:Shab71}.

The inverse scattering method is a standard approach to several completely
integrable nonlinear partial differential equations, including the
Korteweg-de Vries, nonlinear Schr\"{o}dinger or Gross--Pitaevskii,
Sine-Gordon and Kadomtsev--Petviashvili equations \cite{AblowClark91}, \cite%
{Ablo:Seg81}, \cite{Novikovetal}, \cite{PitString03Book}, \cite{Scott03},
and \cite{ScottEncNS}. At the same time, the physical situations in which
these equations arise tend to be highly idealized. The inclusion of effects
such as damping, external forces, an inhomogeneous medium with variable
density and a higher order of the nonlinearity may provide a more realistic
model. However, the addition of these perturbation effects could mean that
the system is no longer completely integrable (see an example in Refs.~\cite%
{AblowClark91}, \cite{Clark88}, \cite{GagWint93}, and \cite{Suslov11}).
Hence it is of interest to determine under what conditions the perturbation
preserves the completely integrability. Some nonintegrable systems in
question posses important in practice classes of exact solutions (see, for
instance, \cite{ConteMusette93}, \cite{ConteMusette09}, \cite{MarcHugConte94}%
, \cite{MusetteConte03}\ and the references therein). These solutions may
serve as a natural starting point of perturbation methods and also provide a
useful testing ground for numerical investigation of more complicated models.

In this paper, we show that the following nonlinear PDE in the
two-dimensional space-time continuum:%
\begin{equation}
\frac{\partial \psi }{\partial t}+Q\left( \frac{\partial }{\partial x}%
,x\right) \psi =P\left( \psi ,\psi ^{\ast },\frac{\partial \psi }{\partial x}%
,\frac{\partial \psi ^{\ast }}{\partial x}\right) ,  \label{GeneralPDE}
\end{equation}%
where $Q$ is a quadratic of two (noncommuting) operators $\partial /\partial
x$ and $x$ with time-dependent coefficients, the asterisk denotes the
complex conjugation, and $P$ is certain gauge invariant nonlinearity, can be
reduced by a simple change of variables to the autonomous form:%
\begin{equation}
\frac{\partial \chi }{\partial \tau }+\frac{\partial ^{2}\chi }{\partial \xi
^{2}}=R\left( \chi ,\chi ^{\ast },\frac{\partial \chi }{\partial \xi }%
,\left( \frac{\partial \chi }{\partial \xi }\right) ^{\ast }\right) ,
\label{GeneralStandard}
\end{equation}%
which is somehow easier to analyze by a variety of available tools (in
general, the new time variable $\tau $ may be complex-valued and the new
space variable $\xi $ is a linear complex-valued function of $x;$ the
precise form of this transformation is given by Theorem~1 below). Among
important in applications examples are the derivative nonlinear Schr\"{o}%
dinger equation, which describes the propagation of circular polarized
nonlinear Alfv\'{e}n waves in plasma physics, the quintic complex
Ginzburg--Landau equation, the Gerdjikov--Ivanov equation, and others. Here,
we mainly concentrate on variants of the nonlinear Schr\"{o}dinger equation
for which the unperturbed models are known to be completely integrable
and/or have explicit solutions.

The paper is organized as follows. In the next section, we transform (\ref%
{GeneralPDE}) into the standard form (\ref{GeneralStandard}). An overview of
integrable autonomous cases and some exact solutions of important
nonintegrable PDE models are given in sections~3 to 6 in order to make our
presentation as self-contain as possible (a detailed bibliography is
provided). Some applications are discussed in section~7. In particular, we
study self-similar and nonspreading wave packets for cubic nonlinear Schr%
\"{o}dinger equations and quintic complex Ginzburg--Landau equations when
some other interesting coherent structures are also available. A
self-focusing wave packet for the linear Schr\"{o}dinger equation, an
\textquotedblleft Airy gun\textquotedblright , which self-accelerates to an
infinite velocity in a finite time, is found in analogy to the well-known
blow up property of nonlinear PDE solutions. A natural hierarchy of these
wave packets is discussed from a \textquotedblleft hidden\textquotedblright\
symmetry view point and applications to inhomogeneous media including fiber
optics, atmospheric plasma, hydrodynamics, and Bose condensation are briefly
reviewed.

\section{Transforming Nonlinear Schr\"{o}dinger Equation into Autonomous
Forms}

We start from the nonautonomous nonlinear Schr\"{o}dinger equation%
\begin{equation}
i\frac{\partial \psi }{\partial t}=H\left( t\right) \psi +h\left( t\right)
\left\vert \psi \right\vert ^{p}\psi  \label{Schroedinger}
\end{equation}%
on $%
\mathbb{R}
,$ where the variable Hamiltonian $H=Q\left( p,x,t\right) $ is an arbitrary
quadratic form of two operators $p=-i\partial /\partial x$ and $x,$ namely,%
\begin{equation}
i\psi _{t}=-a\left( t\right) \psi _{xx}+b\left( t\right) x^{2}\psi -ic\left(
t\right) x\psi _{x}-id\left( t\right) \psi -f\left( t\right) x\psi +ig\left(
t\right) \psi _{x}+h\left( t\right) \left\vert \psi \right\vert ^{p}\psi
\label{SchroedingerQuadratic}
\end{equation}%
($a,$ $b,$ $c,$ $d,$ $f,$ and $g$ are suitable real-valued functions of time
only) under the following integrability condition \cite{SuazoSuslovSol}, 
\cite{Suslov11}:%
\begin{equation}
h=h_{0}a\left( t\right) \beta ^{2}\left( t\right) \mu ^{p/2}\left( t\right)
=h_{0}\beta ^{2}\left( 0\right) \mu ^{2}\left( 0\right) \frac{a\left(
t\right) \lambda ^{2}\left( t\right) }{\left( \mu \left( t\right) \right)
^{2-p/2}}  \label{SolInt2a}
\end{equation}%
($h_{0}$ is a constant, functions $\beta ,$ $\lambda ,$ and $\mu $ will be
defined below).

(The corresponding linear quantum systems, when $h=0,$ are known as the 
\textit{generalized (driven) harmonic oscillators}. Some examples, a general
approach and known elementary solutions can be found in Refs.~\cite%
{Cor-Sot:Lop:Sua:Sus}, \cite{Cor-Sot:Sua:Sus}, \cite{Cor-Sot:Sua:SusInv}, 
\cite{Cor-Sot:Sus}, \cite{Dod:Mal:Man75}, \cite{Lop:Sus}, \cite{Wolf81}, 
\cite{Yeon:Lee:Um:George:Pandey93}, and \cite{Zhukov99}.)

For completeness we combine the results established in \cite{Lan:Lop:Sus}
and \cite{Suslov11}.

\begin{lemma}
The substitution%
\begin{equation}
\psi =\frac{e^{i\left( \alpha \left( t\right) x^{2}+\delta \left( t\right)
x+\kappa \left( t\right) \right) }}{\sqrt{\mu \left( t\right) }}\ \chi
\left( \xi ,\tau \right) ,\qquad \xi =\beta \left( t\right) x+\varepsilon
\left( t\right) ,\quad \tau =\gamma \left( t\right)  \label{Gauge}
\end{equation}%
transforms the non-autonomous and inhomogeneous Schr\"{o}dinger equation (%
\ref{SchroedingerQuadratic}) into the autonomous form%
\begin{equation}
i\chi _{\tau }+h_{0}\left\vert \chi \right\vert ^{p}\chi =\chi _{\xi \xi
}-c_{0}\xi ^{2}\chi \qquad \left( c_{0}=0,1\right)  \label{ASEq}
\end{equation}%
provided that%
\begin{equation}
\frac{d\alpha }{dt}+b+2c\alpha +4a\alpha ^{2}=c_{0}a\beta ^{4},  \label{SysA}
\end{equation}%
\begin{equation}
\frac{d\beta }{dt}+\left( c+4a\alpha \right) \beta =0,  \label{SysB}
\end{equation}%
\begin{equation}
\frac{d\gamma }{dt}+a\beta ^{2}=0  \label{SysC}
\end{equation}%
and%
\begin{equation}
\frac{d\delta }{dt}+\left( c+4a\alpha \right) \delta =f+2g\alpha
+2c_{0}a\beta ^{3}\varepsilon ,  \label{SysD}
\end{equation}%
\begin{equation}
\frac{d\varepsilon }{dt}=\left( g-2a\delta \right) \beta ,  \label{SysE}
\end{equation}%
\begin{equation}
\frac{d\kappa }{dt}=g\delta -a\delta ^{2}+c_{0}a\beta ^{2}\varepsilon ^{2}.
\label{SysF}
\end{equation}%
Here%
\begin{equation}
\alpha =\frac{1}{4a}\frac{\mu ^{\prime }}{\mu }-\frac{d}{2a}.  \label{Alpha}
\end{equation}
\end{lemma}

(A \textsl{Mathematica} based proof of Lemma~1 is given by Christoph
Koutschan \cite{Kouchan11}.)

The substitution (\ref{Alpha}) reduces the inhomogeneous equation (\ref{SysA}%
) to the second order ordinary differential equation:%
\begin{equation}
\mu ^{\prime \prime }-\tau \left( t\right) \mu ^{\prime }+4\sigma \left(
t\right) \mu =c_{0}\left( 2a\right) ^{2}\beta ^{4}\mu ,  \label{CharEq}
\end{equation}%
that has the familiar time-varying coefficients%
\begin{equation}
\tau \left( t\right) =\frac{a^{\prime }}{a}-2c+4d,\qquad \sigma \left(
t\right) =ab-cd+d^{2}+\frac{d}{2}\left( \frac{a^{\prime }}{a}-\frac{%
d^{\prime }}{d}\right) .  \label{TauSigma}
\end{equation}

The time-dependent coefficients $\alpha _{0},$ $\beta _{0},$ $\gamma _{0},$ $%
\delta _{0},$ $\varepsilon _{0},$ $\kappa _{0}$ satisfy the Riccati-type
system (\ref{SysA})--(\ref{SysF}) with $c_{0}=0$ and are given as follows 
\cite{Cor-Sot:Lop:Sua:Sus}, \cite{Suslov10}:%
\begin{eqnarray}
&&\alpha _{0}\left( t\right) =\frac{1}{4a\left( t\right) }\frac{\mu
_{0}^{\prime }\left( t\right) }{\mu _{0}\left( t\right) }-\frac{d\left(
t\right) }{2a\left( t\right) },  \label{A0} \\
&&\beta _{0}\left( t\right) =-\frac{\lambda \left( t\right) }{\mu _{0}\left(
t\right) },\qquad \lambda \left( t\right) =\exp \left( -\int_{0}^{t}\left(
c\left( s\right) -2d\left( s\right) \right) \ ds\right) ,  \label{B0} \\
&&\gamma _{0}\left( t\right) =\frac{1}{2\mu _{1}\left( 0\right) }\frac{\mu
_{1}\left( t\right) }{\mu _{0}\left( t\right) }+\frac{d\left( 0\right) }{%
2a\left( 0\right) }  \label{C0}
\end{eqnarray}%
and%
\begin{equation}
\delta _{0}\left( t\right) =\frac{\lambda \left( t\right) }{\mu _{0}\left(
t\right) }\int_{0}^{t}\left[ \left( f\left( s\right) -\frac{d\left( s\right) 
}{a\left( s\right) }g\left( s\right) \right) \mu _{0}\left( s\right) +\frac{%
g\left( s\right) }{2a\left( s\right) }\mu _{0}^{\prime }\left( s\right) %
\right] \frac{ds}{\lambda \left( s\right) },  \label{D0}
\end{equation}%
\begin{eqnarray}
\varepsilon _{0}\left( t\right) &=&-\frac{2a\left( t\right) \lambda \left(
t\right) }{\mu _{0}^{\prime }\left( t\right) }\delta _{0}\left( t\right)
+8\int_{0}^{t}\frac{a\left( s\right) \sigma \left( s\right) \lambda \left(
s\right) }{\left( \mu _{0}^{\prime }\left( s\right) \right) ^{2}}\left( \mu
_{0}\left( s\right) \delta _{0}\left( s\right) \right) \ ds  \label{E0} \\
&&\quad +2\int_{0}^{t}\frac{a\left( s\right) \lambda \left( s\right) }{\mu
_{0}^{\prime }\left( s\right) }\left( f\left( s\right) -\frac{d\left(
s\right) }{a\left( s\right) }g\left( s\right) \right) \ ds,  \notag
\end{eqnarray}%
\begin{eqnarray}
\kappa _{0}\left( t\right) &=&\frac{a\left( t\right) \mu _{0}\left( t\right) 
}{\mu _{0}^{\prime }\left( t\right) }\delta _{0}^{2}\left( t\right)
-4\int_{0}^{t}\frac{a\left( s\right) \sigma \left( s\right) }{\left( \mu
_{0}^{\prime }\left( s\right) \right) ^{2}}\left( \mu _{0}\left( s\right)
\delta _{0}\left( s\right) \right) ^{2}\ ds  \label{F0} \\
&&\quad -2\int_{0}^{t}\frac{a\left( s\right) }{\mu _{0}^{\prime }\left(
s\right) }\left( \mu _{0}\left( s\right) \delta _{0}\left( s\right) \right)
\left( f\left( s\right) -\frac{d\left( s\right) }{a\left( s\right) }g\left(
s\right) \right) \ ds  \notag
\end{eqnarray}%
$(\delta _{0}\left( 0\right) =-\varepsilon _{0}\left( 0\right) =g\left(
0\right) /\left( 2a\left( 0\right) \right) $ and $\kappa _{0}\left( 0\right)
=0)$ provided that $\mu _{0}$ and $\mu _{1}$ are the standard (real-valued)
solutions of equation (\ref{CharEq}) when $c_{0}=0$ corresponding to the
initial conditions $\mu _{0}\left( 0\right) =0,$ $\mu _{0}^{\prime }\left(
0\right) =2a\left( 0\right) \neq 0$ and $\mu _{1}\left( 0\right) \neq 0,$ $%
\mu _{1}^{\prime }\left( 0\right) =0.$ (Proofs of these facts are outlined
in Refs.~\cite{Cor-Sot:Lop:Sua:Sus} and \cite{Cor-Sot:SusDPO}. Here, the
integrals are treated in the most possible general way which may include
stochastic calculus.)

The systems (\ref{SysA})--(\ref{SysF}) can be solved by the following
variants of a nonlinear superposition principle \cite{Lan:Lop:Sus} and \cite%
{SuazoSuslovSol}.

\begin{lemma}
The solution of the Riccati-type system (\ref{SysA})--(\ref{SysF}) when $%
c_{0}=0$ is given by%
\begin{eqnarray}
&&\mu \left( t\right) =2\mu \left( 0\right) \mu _{0}\left( t\right) \left(
\alpha \left( 0\right) +\gamma _{0}\left( t\right) \right) ,  \label{MKernel}
\\
&&\alpha \left( t\right) =\alpha _{0}\left( t\right) -\frac{\beta
_{0}^{2}\left( t\right) }{4\left( \alpha \left( 0\right) +\gamma _{0}\left(
t\right) \right) },  \label{AKernel} \\
&&\beta \left( t\right) =-\frac{\beta \left( 0\right) \beta _{0}\left(
t\right) }{2\left( \alpha \left( 0\right) +\gamma _{0}\left( t\right)
\right) }=\frac{\beta \left( 0\right) \mu \left( 0\right) }{\mu \left(
t\right) }\lambda \left( t\right) ,  \label{BKernel} \\
&&\gamma \left( t\right) =\gamma \left( 0\right) -\frac{\beta ^{2}\left(
0\right) }{4\left( \alpha \left( 0\right) +\gamma _{0}\left( t\right)
\right) }  \label{CKernel}
\end{eqnarray}%
and%
\begin{eqnarray}
\delta \left( t\right) &=&\delta _{0}\left( t\right) -\frac{\beta _{0}\left(
t\right) \left( \delta \left( 0\right) +\varepsilon _{0}\left( t\right)
\right) }{2\left( \alpha \left( 0\right) +\gamma _{0}\left( t\right) \right) 
},  \label{DKernel} \\
\varepsilon \left( t\right) &=&\varepsilon \left( 0\right) -\frac{\beta
\left( 0\right) \left( \delta \left( 0\right) +\varepsilon _{0}\left(
t\right) \right) }{2\left( \alpha \left( 0\right) +\gamma _{0}\left(
t\right) \right) },  \label{EKernel} \\
\kappa \left( t\right) &=&\kappa \left( 0\right) +\kappa _{0}\left( t\right)
-\frac{\left( \delta \left( 0\right) +\varepsilon _{0}\left( t\right)
\right) ^{2}}{4\left( \alpha \left( 0\right) +\gamma _{0}\left( t\right)
\right) }  \label{FKernel}
\end{eqnarray}%
in terms of the fundamental solution (\ref{A0})--(\ref{F0}) subject to the
arbitrary initial data $\mu \left( 0\right) ,$ $\alpha \left( 0\right) ,$ $%
\beta \left( 0\right) \neq 0,$ $\gamma \left( 0\right) ,$ $\delta \left(
0\right) ,$ $\varepsilon \left( 0\right) ,$ $\kappa \left( 0\right) .$
\end{lemma}

\begin{lemma}
The solution of the Ermakov-type system (\ref{SysA})--(\ref{SysF}) when $%
c_{0}=1\left( \neq 0\right) $ is given by%
\begin{eqnarray}
&&\mu =\mu \left( 0\right) \mu _{0}\sqrt{\beta ^{4}\left( 0\right) +4\left(
\alpha \left( 0\right) +\gamma _{0}\right) ^{2}},  \label{MKernelOsc} \\
&&\alpha =\alpha _{0}-\beta _{0}^{2}\frac{\alpha \left( 0\right) +\gamma _{0}%
}{\beta ^{4}\left( 0\right) +4\left( \alpha \left( 0\right) +\gamma
_{0}\right) ^{2}},  \label{AKernelOsc} \\
&&\beta =-\frac{\beta \left( 0\right) \beta _{0}}{\sqrt{\beta ^{4}\left(
0\right) +4\left( \alpha \left( 0\right) +\gamma _{0}\right) ^{2}}}=\frac{%
\beta \left( 0\right) \mu \left( 0\right) }{\mu \left( t\right) }\lambda
\left( t\right) ,  \label{BKernelOsc} \\
&&\gamma =\gamma \left( 0\right) -\frac{1}{2}\arctan \frac{\beta ^{2}\left(
0\right) }{2\left( \alpha \left( 0\right) +\gamma _{0}\right) },\quad
a\left( 0\right) >0  \label{CKernelOsc}
\end{eqnarray}%
and%
\begin{eqnarray}
&&\delta =\delta _{0}-\beta _{0}\frac{\varepsilon \left( 0\right) \beta
^{3}\left( 0\right) +2\left( \alpha \left( 0\right) +\gamma _{0}\right)
\left( \delta \left( 0\right) +\varepsilon _{0}\right) }{\beta ^{4}\left(
0\right) +4\left( \alpha \left( 0\right) +\gamma _{0}\right) ^{2}},
\label{DKernelOsc} \\
&&\varepsilon =\frac{2\varepsilon \left( 0\right) \left( \alpha \left(
0\right) +\gamma _{0}\right) -\beta \left( 0\right) \left( \delta \left(
0\right) +\varepsilon _{0}\right) }{\sqrt{\beta ^{4}\left( 0\right) +4\left(
\alpha \left( 0\right) +\gamma _{0}\right) ^{2}}},  \label{EKernelOsc} \\
&&\kappa =\kappa \left( 0\right) +\kappa _{0}-\varepsilon \left( 0\right)
\beta ^{3}\left( 0\right) \frac{\delta \left( 0\right) +\varepsilon _{0}}{%
\beta ^{4}\left( 0\right) +4\left( \alpha \left( 0\right) +\gamma
_{0}\right) ^{2}}  \label{FKernelOsc} \\
&&\qquad +\left( \alpha \left( 0\right) +\gamma _{0}\right) \frac{%
\varepsilon ^{2}\left( 0\right) \beta ^{2}\left( 0\right) -\left( \delta
\left( 0\right) +\varepsilon _{0}\right) ^{2}}{\beta ^{4}\left( 0\right)
+4\left( \alpha \left( 0\right) +\gamma _{0}\right) ^{2}}  \notag
\end{eqnarray}%
in terms of the fundamental solution (\ref{A0})--(\ref{F0}) subject to the
arbitrary initial data $\mu \left( 0\right) ,$ $\alpha \left( 0\right) ,$ $%
\beta \left( 0\right) \neq 0,$ $\gamma \left( 0\right) ,$ $\delta \left(
0\right) ,$ $\varepsilon \left( 0\right) ,$ $\kappa \left( 0\right) .$
\end{lemma}

Here we would like to present a new compact form of these solutions. In
order to do that, let us introduce the following complex-valued function:%
\begin{equation}
z=\left( 2\alpha \left( 0\right) +\frac{d\left( 0\right) }{a\left( 0\right) }%
\right) \mu _{0}\left( t\right) +\frac{\mu _{1}\left( t\right) }{\mu
_{1}\left( 0\right) }+i\beta ^{2}\left( 0\right) \mu _{0}\left( t\right)
\label{ComplexZ}
\end{equation}%
(a complex parametrization of Green's function and linear invariants of
generalized harmonic oscillators are also discussed in Refs.~\cite%
{Dodonov:Man'koFIAN87} and \cite{Har:Ben-Ar:Mann11}). Then%
\begin{equation}
z=c_{1}E\left( t\right) +c_{2}E^{\ast }\left( t\right) ,  \label{ComplexZE}
\end{equation}%
where the complex-valued solutions are given by%
\begin{equation}
E\left( t\right) =\frac{\mu _{1}\left( t\right) }{\mu _{1}\left( 0\right) }%
+i\mu _{0}\left( t\right) ,\qquad E^{\ast }\left( t\right) =\frac{\mu
_{1}\left( t\right) }{\mu _{1}\left( 0\right) }-i\mu _{0}\left( t\right)
\label{ComplexE}
\end{equation}%
and the corresponding complex-valued parameters are defined as follows 
\begin{equation}
c_{1}=\frac{1+\beta ^{2}\left( 0\right) }{2}-i\left( \alpha \left( 0\right) +%
\frac{d\left( 0\right) }{2a\left( 0\right) }\right) ,\qquad c_{2}=\frac{%
1-\beta ^{2}\left( 0\right) }{2}+i\left( \alpha \left( 0\right) +\frac{%
d\left( 0\right) }{2a\left( 0\right) }\right)  \label{ComplexC12}
\end{equation}%
with%
\begin{equation}
c_{1}+c_{2}=1,\qquad \left\vert c_{1}\right\vert ^{2}-\left\vert
c_{2}\right\vert ^{2}=c_{1}-c_{2}^{\ast }=\beta ^{2}\left( 0\right) .
\label{ComplexC12Relations}
\end{equation}%
In addition,%
\begin{equation}
z\left( 0\right) =c_{1}+c_{2}=1,\qquad z^{\prime }\left( 0\right) =2ia\left(
0\right) \left( c_{1}-c_{2}\right) .  \label{ComplInitData}
\end{equation}%
The inverses are given by%
\begin{equation}
E=\frac{c_{1}^{\ast }z-c_{2}z^{\ast }}{\left\vert c_{1}\right\vert
^{2}-\left\vert c_{2}\right\vert ^{2}},\qquad E^{\ast }=\frac{c_{1}z^{\ast
}-c_{2}^{\ast }z}{\left\vert c_{1}\right\vert ^{2}-\left\vert
c_{2}\right\vert ^{2}}  \label{EZ}
\end{equation}%
and%
\begin{equation}
\mu _{0}=\frac{z-z^{\ast }}{2i\left( c_{1}-c_{2}^{\ast }\right) },\qquad 
\frac{\mu _{1}}{\mu _{1}\left( 0\right) }=\frac{\left( c_{1}^{\ast
}-c_{2}^{\ast }\right) z+\left( c_{1}-c_{2}\right) z^{\ast }}{2\left(
c_{1}-c_{2}^{\ast }\right) }  \label{MZ}
\end{equation}%
in terms of our complex vector (\ref{ComplexZ}).

One can readily verify that%
\begin{eqnarray}
&&\alpha _{0}=\frac{1}{4a}\frac{\left( z-z^{\ast }\right) ^{\prime }}{%
z-z^{\ast }}+\frac{d}{2a},\qquad \beta _{0}=-2i\lambda \frac{%
c_{1}-c_{2}^{\ast }}{z-z^{\ast }},  \label{ABC0} \\
&&\gamma _{0}=\frac{\left( c_{1}^{\ast }-c_{2}^{\ast }\right) z+\left(
c_{1}-c_{2}\right) z^{\ast }}{2i\left( z-z^{\ast }\right) }+\frac{d\left(
0\right) }{2a\left( 0\right) }  \notag
\end{eqnarray}%
and equations (\ref{D0})--(\ref{F0}) can be rewritten in terms of vector $z$
in view of (\ref{MZ}).

Finally, we introduce the second complex vector:%
\begin{equation}
\zeta =\varepsilon \left( 0\right) \beta \left( 0\right) +i\left( \delta
\left( 0\right) +\varepsilon _{0}\right) =c_{3}+i\varepsilon _{0},\qquad
c_{3}=\varepsilon \left( 0\right) \beta \left( 0\right) +i\delta \left(
0\right)  \label{ComplexZeta}
\end{equation}%
and indicate the inverse relations between the essential, real and complex,
parameters:%
\begin{equation}
\alpha \left( 0\right) =\frac{c_{1}^{\ast }-c_{1}}{2i}-\frac{d\left(
0\right) }{2a\left( 0\right) },\qquad \beta ^{2}\left( 0\right)
=c_{1}-c_{2}^{\ast }=\left\vert c_{1}\right\vert ^{2}-\left\vert
c_{2}\right\vert ^{2}  \label{ZAB}
\end{equation}%
and%
\begin{equation}
\delta \left( 0\right) =\frac{c_{3}-c_{3}^{\ast }}{2i},\qquad \varepsilon
\left( 0\right) =\pm \frac{c_{3}+c_{3}^{\ast }}{2\sqrt{\left\vert
c_{1}\right\vert ^{2}-\left\vert c_{2}\right\vert ^{2}}}.  \label{ZDE}
\end{equation}%
Then solutions of the initial value problem for the Riccati and Ermakov-type
systems can be found in the following complex forms.

\begin{lemma}
In terms of our time-dependent complex vectors $\zeta $ and $z,$ the
solution of Riccati-type system is given by: $\mu =\mu \left( 0\right) \ 
\func{Re}z=$ $\mu \left( 0\right) \ \left( z+z^{\ast }\right) /2$ and%
\begin{eqnarray}
&&\alpha =\alpha _{0}-\frac{\lambda ^{2}}{2\left\vert z\right\vert ^{2}}%
\frac{\func{Im}z}{\func{Re}z}=\alpha _{0}-\frac{\lambda ^{2}}{2i\left\vert
z\right\vert ^{2}}\frac{z-z^{\ast }}{z+z^{\ast }},\quad \beta =\pm 2\lambda 
\frac{\sqrt{\left\vert c_{1}\right\vert ^{2}-\left\vert c_{2}\right\vert ^{2}%
}}{z+z^{\ast }}, \\
&&\gamma =\gamma \left( 0\right) -\frac{1}{2i}\frac{z-z^{\ast }}{z+z^{\ast }}%
,\qquad \delta =\delta _{0}+\lambda \frac{\func{Im}\zeta }{\func{Re}z}%
=\delta _{0}-i\lambda \frac{\zeta -\zeta ^{\ast }}{z+z^{\ast }}, \\
&&\varepsilon =\varepsilon \left( 0\right) \pm i\sqrt{\left\vert
c_{1}\right\vert ^{2}-\left\vert c_{2}\right\vert ^{2}}\ \frac{\zeta -\zeta
^{\ast }}{z+z^{\ast }},\quad \kappa =\kappa \left( 0\right) +\kappa _{0}+%
\frac{\left( \zeta -\zeta ^{\ast }\right) ^{2}}{8i\left( c_{1}-c_{2}^{\ast
}\right) }\ \frac{z-z^{\ast }}{z+z^{\ast }}.
\end{eqnarray}
\end{lemma}

\begin{lemma}
The Ermakov-type system can be solved as follows: $\mu =\mu \left( 0\right)
\ \left\vert z\right\vert $ and%
\begin{eqnarray}
&&\alpha =\alpha _{0}+\lambda ^{2}\frac{c_{1}-c_{2}^{\ast }}{2i\left\vert
z\right\vert ^{2}}\ \frac{z+z^{\ast }}{z-z^{\ast }},\quad \beta =\pm \lambda 
\frac{\sqrt{\left\vert c_{1}\right\vert ^{2}-\left\vert c_{2}\right\vert ^{2}%
}}{\left\vert z\right\vert },\quad \gamma =\gamma \left( 0\right) -\frac{1}{2%
}\arg z,  \label{ABC} \\
&&\delta =\delta _{0}+\lambda \frac{\zeta z-\zeta ^{\ast }z^{\ast }}{%
2i\left\vert z\right\vert ^{2}},\qquad \varepsilon =\pm \frac{\zeta z+\zeta
^{\ast }z^{\ast }}{2\left\vert z\right\vert \sqrt{\left\vert
c_{1}\right\vert ^{2}-\left\vert c_{2}\right\vert ^{2}}},  \label{CE} \\
&&\kappa =\kappa \left( 0\right) +\kappa _{0}+\frac{\left( \zeta
^{2}z+\left. \zeta ^{\ast }\right. ^{2}z^{\ast }\right) \left( z-z^{\ast
}\right) }{8i\left( c_{1}-c_{2}^{\ast }\right) \left\vert z\right\vert ^{2}}.
\label{F}
\end{eqnarray}
\end{lemma}

(The proofs are rather straightforward and left to the reader.)

As a consequence, for the Ermakov-type system one gets%
\begin{eqnarray}
&&\qquad \qquad \qquad \qquad \frac{2i\left( \alpha -\alpha _{0}\right) }{%
\beta ^{2}}=\frac{z+z^{\ast }}{z-z^{\ast }}, \\
&&i\left( \alpha -\alpha _{0}\right) +\frac{\beta ^{2}}{2}=\beta ^{2}\frac{z%
}{z-z^{\ast }},\qquad i\left( \alpha -\alpha _{0}\right) -\frac{\beta ^{2}}{2%
}=\beta ^{2}\frac{z^{\ast }}{z-z^{\ast }}
\end{eqnarray}%
\begin{eqnarray}
&&\varepsilon +i\frac{\delta -\delta _{0}}{\beta }=\frac{\zeta z}{\beta
\left( 0\right) \left\vert z\right\vert },\qquad \varepsilon -i\frac{\delta
-\delta _{0}}{\beta }=\frac{\zeta ^{\ast }z^{\ast }}{\beta \left( 0\right)
\left\vert z\right\vert }, \\
&&\qquad \varepsilon ^{2}+\left( \frac{\delta -\delta _{0}}{\beta }\right)
^{2}=\varepsilon ^{2}\left( 0\right) +\left( \frac{\delta \left( 0\right)
+\varepsilon _{0}}{\beta \left( 0\right) }\right) ^{2}
\end{eqnarray}%
and%
\begin{equation}
\kappa =\kappa \left( 0\right) +\kappa _{0}+\frac{\delta -\delta _{0}}{%
2\beta }\varepsilon -\frac{\varepsilon _{0}+\delta \left( 0\right) }{2\beta
\left( 0\right) }\varepsilon \left( 0\right) .
\end{equation}%
These \textquotedblleft quasi-invariants\textquotedblright\ can be useful,
for example, when making comparison of calculations done by different
approximation methods.

In this paper, we extend \textquotedblleft The Simple Three Lemmas
Approach\textquotedblright , which is summarized and modified above for the
linear problem, by the complexification of all time-dependent coefficients
and by a classification of the simplest nonautonomous nonlinear terms. For a
generic nonautonomous (derivative) nonlinear Schr\"{o}dinger equation of the
form 
\begin{equation}
i\psi _{t}=H\psi +R\left( \psi \right) ,  \label{DNLSE}
\end{equation}%
where $H$ is an arbitrary variable quadratic Hamiltonian and $R\left( \psi
\right) =P\left( \psi ,\psi ^{\ast },\psi _{x},\psi _{x}^{\ast }\right) $ is
a polynomial in four variables, we assume the natural gauge invariance
condition:%
\begin{equation}
P\left( \psi e^{iS},\psi ^{\ast }e^{-iS},\left( \psi e^{iS}\right)
_{x},\left( \psi ^{\ast }e^{-iS}\right) _{x}\right) =Ce^{iS}P\left( \psi
,\psi ^{\ast },\psi _{x},\psi _{x}^{\ast }\right) .
\label{DNLSENonlinearity}
\end{equation}%
The lowest terms that satisfy this condition are given by%
\begin{equation}
P\left( \psi ,\psi ^{\ast },\psi _{x},\psi _{x}^{\ast }\right) =h_{0}\psi
+\left( h_{1}x+h_{2}\right) \left\vert \psi \right\vert ^{2}\psi
+ih_{3}\left\vert \psi \right\vert ^{2}\psi _{x}+ih_{4}\psi ^{2}\psi
_{x}^{\ast }+h_{5}\left\vert \psi \right\vert ^{4}\psi
\label{DNLSENonlinearity5}
\end{equation}%
where $h_{k}=h_{k}\left( x,t\right) ,$ $k=0,$ $...,$ $5$ are some real or
complex-valued functions. Our result is as follows.

\begin{theorem}
The substitution (\ref{Gauge}) transforms the non-autonomous and
inhomogeneous equation (\ref{DNLSE}) with the lowest gauge invariant
nonlinearities (\ref{DNLSENonlinearity5}), namely,%
\begin{eqnarray}
&&i\psi _{t}+a\left( t\right) \psi _{xx}-b\left( t\right) x^{2}\psi
+ic\left( t\right) x\psi _{x}+id\left( t\right) \psi +f\left( t\right) x\psi
-ig\left( t\right) \psi _{x}  \label{DNLSEExplicit} \\
&&\qquad =h_{0}\psi +\left( h_{1}x+h_{2}\right) \left\vert \psi \right\vert
^{2}\psi +ih_{3}\left\vert \psi \right\vert ^{2}\psi _{x}+ih_{4}\psi
^{2}\psi _{x}^{\ast }+h_{5}\left\vert \psi \right\vert ^{4}\psi ,  \notag
\end{eqnarray}%
into the autonomous form%
\begin{eqnarray}
-i\chi _{\tau }+\chi _{\xi \xi }-c_{0}\xi ^{2}\chi &=&d_{0}\chi +\left(
d_{1}\xi +d_{2}\right) \left\vert \chi \right\vert ^{2}\chi
+id_{3}\left\vert \chi \right\vert ^{2}\chi _{\xi }+id_{4}\chi ^{2}\left(
\chi _{\xi }\right) ^{\ast }+d_{5}\left\vert \chi \right\vert ^{4}\chi \qquad
\label{DNLSEStandardForm} \\
&&\left( c_{0}=0,1\right)  \notag
\end{eqnarray}%
provided that%
\begin{eqnarray}
&&h_{1}=a\beta ^{2}\left\vert \mu \right\vert \left( d_{1}\beta +\frac{%
2\alpha }{\beta }d_{3}-\frac{2\alpha ^{\ast }}{\beta ^{\ast }}d_{4}\right)
e^{2\func{Im}S},  \label{H1} \\
&&h_{2}=a\beta ^{2}\left\vert \mu \right\vert \left[ d_{1}\varepsilon +d_{2}+%
\frac{\delta }{\beta }d_{3}-\frac{\delta ^{\ast }}{\beta ^{\ast }}d_{4}%
\right] e^{2\func{Im}S},  \label{H2} \\
&&h_{3}=d_{3}a\beta \left\vert \mu \right\vert e^{2\func{Im}S},\qquad
h_{4}=d_{4}a\frac{\beta ^{2}}{\beta ^{\ast }}\left\vert \mu \right\vert e^{2%
\func{Im}S},  \label{H34} \\
&&h_{5}=d_{5}a\beta ^{2}\left\vert \mu \right\vert ^{2}e^{4\func{Im}%
S},\qquad h_{0}=d_{0}a\beta ^{2}.  \label{H50}
\end{eqnarray}%
Here $d_{0},$ $d_{1},$ $d_{2},$ $d_{3},$ $d_{4},$ $d_{5}$ are constants and $%
S=\alpha x^{2}+\delta x+\kappa .$
\end{theorem}

\begin{proof}
For the case of complex-valued coefficients, the transformation of the
linear part is similar to \cite{Lan:Lop:Sus} and \cite{Suslov11} (with the
use of contour integration if needed). Changing the variables in the
nonlinear part:%
\begin{eqnarray}
&&\mu ^{1/2}e^{-iS}P \\
&&\quad =h_{0}\chi +ih_{3}\frac{\beta }{\left\vert \mu \right\vert }e^{-2%
\func{Im}S}\left\vert \chi \right\vert ^{2}\chi _{\xi }+ih_{4}\frac{\beta
^{\ast }}{\left\vert \mu \right\vert }e^{-2\func{Im}S}\chi ^{2}\left( \chi
_{\xi }\right) ^{\ast }+\frac{h_{5}}{\left\vert \mu \right\vert ^{2}}e^{-4%
\func{Im}S}\left\vert \chi \right\vert ^{4}\chi  \notag \\
&\quad &\qquad +\left[ h_{2}-h_{3}\delta +h_{4}\delta ^{\ast }+\left(
h_{1}-2\alpha h_{3}+2\alpha ^{\ast }h_{4}\right) \frac{\xi -\varepsilon }{%
\beta }\right] e^{-2\func{Im}S}\frac{\left\vert \chi \right\vert ^{2}}{%
\left\vert \mu \right\vert }\chi  \notag
\end{eqnarray}%
and substituting into (\ref{DNLSEExplicit}), one completes the proof with
the aid of our conditions (\ref{H1})--(\ref{H50}).
\end{proof}

When $c_{0}=d_{0}=d_{1}=d_{3}=d_{4}=d_{5}=0,$ $d_{2}=h_{0}\neq 0$ and $%
d_{0}=d_{1}=d_{2}=d_{3}=d_{4}=0,$ $d_{5}=h_{0}\neq 0,$ we reproduce the
results of Lemma~1 for $p=2$ and $p=4,$ respectively. The derivative
nonlinear Schr\"{o}dinger equation, which was first derived for the
propagation of circular polarized nonlinear Alfv\'{e}n waves in plasma
physics \cite{Mioetal76}, \cite{Mioetal76a}, and \cite{Wadatietal78},
appears when $c_{0}=d_{0}=d_{1}=d_{2}=d_{5}=0$ and $d_{3}=2d_{4}$ (see \cite%
{AbloRamSegII}, \cite{ChenLeeLiu79}, \cite{ClarkCosgr87}, \cite{KaupNewell78}%
, \cite{KawKobIno79}, \cite{KawIno78}, \cite{NiaNing07} and the references
therein for methods of solution of this equation). The amplitude equations
of cubic-quintic type have been derived via asymptotic analysis of the
governing Navier-Stokes equations of fluid mechanics in the limit of long
spatial and slow temporal oscillations near the onset of instability \cite%
{NewellSIAM}. Among other special cases are the Chen--Lee--Lui derivative
nonlinear Schr\"{o}dinger equation \cite{ChenLeeLiu79} and the
Gerdjikov--Ivanov equation \cite{GerdIvan83}, \cite{YuHeHan12}.

Generic autonomous equations of the type (\ref{DNLSEStandardForm}) and some
of their extensions such as quintic complex Ginzburg--Landau equation are
discussed in \cite{BrazhKonotopPita06}, \cite{CalogeroEckhaus87}, \cite%
{Clark88}, \cite{Clark92}, \cite{ClarkCosgr87}, \cite{Gagnon89}, \cite%
{GagGramRamWint89}, \cite{GagWint88}, \cite{GagWint89}, \cite{GagWint93}, 
\cite{Johnson77}, \cite{KakSasSat95}, \cite{KengneLiu06}, \cite{Kohler02}, 
\cite{Kundu94}, \cite{Kundu06}, \cite{MarcHugConte94}, \cite{MoralesLee78}, 
\cite{MuryShlyapetal02}, \cite{Peletal96}, \cite{Saa03}, \cite{SaaHoh92}, 
\cite{Wadatietal79a}, and \cite{WadatiSogo83} (see also the references
therein).

If conditions (\ref{H1})--(\ref{H50}) are not satisfied in certain
application, one yet can find $c_{k}$ as functions of time and spatial
variables thus moving the time dependence into the nonlinear part only. Then
one may use perturbation methods and/or some parameter control when
possible. The Feshbach resonance in Bose--Einstein condensation provides a
classical example of such nonlinearity control \cite%
{Dal:Giorg:Pitaevski:Str99}, \cite{KaganSurShlyap97PRL}, \cite{Pit98}, \cite%
{Pit06}, \cite{PitString03Book}, \cite{Pollacketal10}, and \cite%
{PollacketalHulet09}. A justification of the so-called local density
approximation (see \cite{BrazhKonotopPita06}, \cite{Kolomeiskyetal00} and
the references therein) can be obtained from Lemma~1, $c_{0}=0,$ under the
following adiabatic condition:%
\begin{equation}
\frac{d}{dt}\left( \frac{h}{a\left( t\right) \beta ^{2}\left( t\right) \mu
^{p/2}\left( t\right) }\right) \ll 1,  \label{Adiabatic}
\end{equation}%
when a classical motion of the corresponding quadratic system is already
taken into account (in general, the derivative should be taken with respect
to a small parameter of the system under consideration).

Moreover, in some important autonomous special cases, our Theorem~1 allows
to determine the maximum symmetry groups of the particular equations; see,
for example, \cite{Niederer72}, \cite{Niederer73}, \cite{Lop:Sus:VegaGroup}
and the references therein. An extension to random-valued coefficients $f$
and $g$ is also possible, cf. \cite{Flan:Mah12}. Certain variations of
initial data can be analyzed too.

\section{Integrability of Generalized Autonomous Nonlinear Schr\"{o}dinger
Equations}

In the case $c_{0}=d_{0}=d_{1}=0,$ a detailed Painlev\'{e} analysis of
equation (\ref{DNLSEStandardForm}) is performed by Clarkson and Cosgrove~%
\cite{ClarkCosgr87} (see also \cite{Clark92} for the extension to the case
of complex parameters). They have shown that this equation possesses the
Painlev\'{e} property for partial differential equations only if $%
d_{5}=d_{4}(2d_{4}-d_{3})/4.$ When this relation holds, this is equivalent
under a gauge transformation \cite{Kundu94}, \cite{KakSasSat95}:%
\begin{equation}
\chi =\phi \exp \left( -i\nu \int_{\xi _{0}}^{\xi }\left\vert \phi
\right\vert ^{2}\ d\xi \right) ,  \label{GaugeDNLSE}
\end{equation}%
to a hybrid of the nonlinear Schr\"{o}dinger equation and the derivative
nonlinear Schr\"{o}dinger equation,%
\begin{equation}
-i\phi _{\tau }+\phi _{\xi \xi }+\lambda \left\vert \phi \right\vert
^{2}\phi +i\mu \left( \left\vert \phi \right\vert ^{2}\phi \right) _{\xi }=0
\label{HybridDNLSE}
\end{equation}%
($\lambda ,$ $\mu $ and $\nu $ are constants), where one can assume that $%
\lambda =0$ without loss of generality \cite{WadatiSogo83}. The
corresponding derivative nonlinear Schr\"{o}dinger equation is known to be
completely integrable \cite{AbloRamSegII}, \cite{Ablo:Seg81}, \cite%
{ChenLeeLiu79}, \cite{KaupNewell78}, \cite{KakSasSat95}, \cite{KawIno78}, 
\cite{Lee84}, \cite{Lee89}, and \cite{NiaNing07} (see also the references
therein). Then solutions of the original equation are constructed by the
gauge transformation (\ref{GaugeDNLSE}) in principle (see, for example,
Refs.~\cite{ClarkCosgr87}, \cite{Clark92}, and \cite{KakSasSat95} for more
details). Explicit solitary wave solutions can be found in \cite{FengWang01}%
, \cite{KakSasSat95}, \cite{KawSakKob80}, \cite{KengneLiu06}, \cite%
{Lashkin07}, \cite{MarcHugConte94}, and \cite{SaaHoh92} (see also the
references therein). Our Theorem~1 allows to extend these results to a
larger class of nonautonomous and inhomogeneous nonlinear Schr\"{o}dinger
equations. In the next sections, we will apply these general results to a
nonautonomous solitons and other coherent structures in inhomogeneous medias
including atmospheric plasmas.

\section{Cubic Nonlinear Schr\"{o}dinger Equation}

The cubic nonlinear Schr\"{o}dinger equation has a plane wave solution with
an amplitude dependent phase. Under a simple condition (Lighthill's
criterion), the modulation instability leads to self-focusing, namely, the
parts of the wave front with the larger amplitude will propagate slower and
this effect will enhance itself by the focusing from adjacent parts
(Benjamin--Feir instability \cite{Bouchbinder03}, \cite{ScottEncNS}, \cite%
{SulemSulem99}, \cite{Tajiri05}, and \cite{Whitham74}). If the pressure of
waves can balance the attraction force, the new stationary solutions exist,
which cannot be obtained by small nonlinear perturbation of the linear
problem. These solitary waves, preserve their shapes during the time
evolution. If the nonlinearity cannot balance the dispersion, a finite time
collapse or dispersion that leads to decay occur. A detailed study of the
exact solutions by different methods and many applications of the nonlinear
Schr\"{o}dinger equations can be found in \cite{AblowClark91}, \cite%
{Ablo:Seg81}, \cite{ChrisSorScottLNPh}, \cite{Dal:Giorg:Pitaevski:Str99}, 
\cite{GagWint89}, \cite{HasegawaTappert73a}, \cite{HasegawaTappert73b}, \cite%
{Kagan:Surkov:Shlyap96}, \cite{Kagan:Surkov:Shlyap97}, \cite%
{Kivsh:Alex:Tur01}, \cite{KudryashovBook10}, \cite{Kundu09}, \cite{Newell78}%
, \cite{Novikovetal}, \cite{Per-G:Tor:Mont}, \cite{SuazoSuslovSol}, and \cite%
{Suslov11} (see also the references therein).

Equation (\ref{ASEq}) is integrable when $c_{0}=0$ and $p=2$ \cite%
{Zakh:Shab71}, \cite{ZakhShab74}, and \cite{ZakhShab79}. It has two standard
forms, namely, focusing and defocusing. In this section, we concentrate on
nonspreading and self-similar solitary wave solutions.

\subsection{A Similarity Reduction to the Second Painlev\'{e}}

In the traditional notations, the corresponding defocusing and focusing
nonlinear Schr\"{o}dinger equations,%
\begin{equation}
i\psi _{t}+\psi _{xx}=\pm 2\left\vert \psi \right\vert ^{2}\psi
\label{DefFoc}
\end{equation}%
by the following substitution%
\begin{equation}
\psi \left( x,t\right) =e^{ig\left( x-2gt^{2}/3\right) t}\ g^{1/3}F\left(
g^{1/3}\left( x-gt^{2}\right) \right) ,\qquad g=a/2  \label{FirstRed}
\end{equation}%
($a$ is the acceleration) can be reduced to the (modified) second Painlev%
\'{e} equations:%
\begin{equation}
F^{\prime \prime }=zF\pm 2F^{3},\quad z=g^{1/3}\left( x-gt^{2}\right) ,
\label{PII}
\end{equation}%
whose (bounded) solutions are the nonlinear Airy functions with known
asymptotics as $z\rightarrow \pm \infty $ \cite{Clark10}, \cite%
{SuazoSuslovSol}.

Combining with the familiar Galilei transformation \cite{Scott03},%
\begin{equation}
\psi \left( x,t\right) =e^{i\left( x-vt/2\right) v/2}\ \chi \left(
x-vt,t\right)  \label{SecondRed}
\end{equation}%
($v$ is the velocity), one obtains a more general solution of this type%
\begin{eqnarray}
\psi \left( x,t\right) &=&e^{i\left( x-vt/2\right) v/2+ig\left(
x-vt-2gt^{2}/3\right) t}  \notag \\
&&\times g^{1/3}F\left( g^{1/3}\left( x-vt-gt^{2}\right) \right)
\label{GereralRed}
\end{eqnarray}%
in terms of the second Painlev\'{e} Airy functions. (A similarity reduction
of the nonlinear Schr\"{o}dinger equation to the second Painlev\'{e}
equation was also discussed in Refs.~\cite{GagWint89}, \cite{GiaJos89}, \cite%
{Smith76}, and \cite{Tajiri83}. The linear case is investigated in Ref.~\cite%
{BerryBalazs79} where self-acceletating Airy beams were introduced.) With
the help of Lemma~1 we extend these results to the most general
nonautonomous and inhomogeneous integrable system of this kind.

\subsection{Self-Similar Solutions}

The substitution%
\begin{equation}
\Psi \left( X,T\right) =\chi \left( \sqrt{\dfrac{2}{\mp h_{0}}}X,\dfrac{2}{%
\pm h_{0}}T\right)  \label{SubXT}
\end{equation}%
transforms equation (\ref{ASEq}) into the focusing and defocusing forms in
the new coordinates $X$ and $T:$%
\begin{equation}
i\Psi _{T}+\Psi _{XX}\pm 2\left\vert \Psi \right\vert ^{2}\Psi =0.
\label{StNonLinSchrEq}
\end{equation}%
By Lemma~1, these transformations give a similarity reduction of the
nonautonomous nonlinear Schr\"{o}dinger equation (\ref{SchroedingerQuadratic}%
) to the second Painlev\'{e} equations (\ref{PII}). As the result,%
\begin{equation}
\psi \left( x,t\right) =\frac{1}{\sqrt{\mu \left( t\right) }}e^{i\left(
\alpha \left( t\right) x^{2}+\delta \left( t\right) x+\kappa \left( t\right)
\right) }\ \Psi \left( \sqrt{\dfrac{\mp h_{0}}{2}}\left( \beta \left(
t\right) x+\varepsilon \left( t\right) \right) ,\dfrac{\pm h_{0}}{2}\gamma
\left( t\right) \right) ,  \label{PsiXT}
\end{equation}%
where $\Psi \left( X,T\right) $ is the solution, say in terms of the second
Painlev\'{e} transcendent in (\ref{GereralRed}), with a trivial change of
the notation.

As is known, the nonautonomous Schr\"{o}dinger equation (\ref%
{SchroedingerQuadratic}) under the integrability condition (\ref{SolInt2a})
has also the following solution:%
\begin{eqnarray}
\psi \left( x,t\right) &=&\frac{e^{i\phi }}{\sqrt{\mu }}\exp \left( i\left(
\alpha x^{2}+\beta xy+\gamma \left( y^{2}-g_{0}\right) +\delta x+\varepsilon
y+\kappa \right) \right)  \label{OneSolSol} \\
&&\times G\left( \beta x+2\gamma y+\varepsilon \right) ,  \notag
\end{eqnarray}%
where the elliptic function $G$ satisfies equation%
\begin{equation}
\left( \frac{dG}{dz}\right) ^{2}=C_{0}+g_{0}G^{2}+\frac{1}{2}%
h_{0}G^{4}\qquad \left( C_{0}\text{ is a constant of integration}\right) .
\label{EllipticFuncs}
\end{equation}%
and $\phi ,$ $y,$ $g_{0}$ and $h_{0}$ are real parameters (see also Ref.~%
\cite{SuazoSuslovSol} for a direct derivation of this solution). Examples
include bright and dark solitons, and Jacobi elliptic transcendental
solutions for nonlinear wave profiles \cite{AblowClark91}, \cite%
{KudryashovBook10}, \cite{Novikovetal}, \cite{Scott03}, \cite{SuazoSuslovSol}%
. In the original case (\ref{DefFoc}), setting $C_{0}=y=0,$ gives the
stationary breather, which is located about $x=0$ and oscillates at a
frequency equal to $g_{0}$ \cite{SatYaj74}, \cite{Scott03}.

\subsection{Asymptotics and Connection Problems}

In the defocusing case, the nonlinear Schr\"{o}dinger equation (\ref{DefFoc}%
) has the bounded solution,%
\begin{equation}
\psi \left( x,t\right) =e^{i\left( x-vt/2\right) v/2+ig\left(
x-vt-2gt^{2}/3\right) t}g^{1/3}A_{k_{0}}\left( g^{1/3}\left(
x-vt-gt^{2}\right) \right) ,  \label{NonlinearAiry}
\end{equation}%
in terms of the nonlinear Airy function $A_{k_{0}},$ when $-1<k_{0}<1$ and $%
k_{0}\neq 0.$ The corresponding asymptotics are investigated in \cite%
{AbloSeg77}, \cite{Ablo:Seg81}, \cite{Bassometal98}, \cite{Clark10}, \cite%
{ClarkMcLeod88}, \cite{Conte99}, \cite{ConteMusette09}, \cite{DeiftZhou93}, 
\cite{DeiftZhou95}, \cite{Miles78}, \cite{Miles80}, \cite{Rosales78}, \cite%
{SegurAb81}, and \cite{Takei02} (see also the references therein for study
of this nonlinear Airy function). The end results are%
\begin{equation}
A_{k_{0}}\left( z\right) =\left\{ 
\begin{array}{c}
k_{0}\text{Ai\/}\left( z\right) ,\qquad z\rightarrow +\infty \bigskip \\ 
r\left\vert z\right\vert ^{-1/4}\sin \left( s\left( z\right) -\theta
_{0}\right) +\text{o}\left( \left\vert z\right\vert ^{-1/4}\right) ,\quad
z\rightarrow -\infty \bigskip .%
\end{array}%
\right.  \label{AsPII}
\end{equation}%
(the parameter $k_{0}$ represents amplitude of this nonlinear wave). Here,
the following relations hold%
\begin{equation}
s\left( z\right) =\frac{2}{3}\left\vert z\right\vert ^{3/2}-\frac{3}{4}%
r^{2}\ln \left\vert z\right\vert ,\qquad r^{2}=-\pi ^{-1}\ln \left(
1-k_{0}^{2}\right)  \label{AsPIIs}
\end{equation}%
and%
\begin{equation}
\theta _{0}=\frac{3}{2}r^{2}\ln 2+\arg \Gamma \left( 1-\frac{i}{2}%
r^{2}\right) +\frac{\pi }{4}\left( 1-2\text{sign\/}\left( k_{0}\right)
\right) .  \label{AsPIIph}
\end{equation}

If $\left\vert k_{0}\right\vert =1,$ then%
\begin{equation}
F\left( z\right) \sim \text{sign\/}\left( k_{0}\right) \sqrt{\left\vert
z\right\vert /2},\qquad z\rightarrow -\infty .  \label{AsPIIk1}
\end{equation}%
If $\left\vert k_{0}\right\vert >1,$ then $F\left( z\right) $ has a pole at
a finite point $z=c_{0},$ dependent on $k_{0},$ and%
\begin{equation}
F\left( z\right) \sim \text{sign\/}\left( k_{0}\right) \left( z-c_{0}\right)
^{-1},\qquad z\rightarrow c_{0}^{+}.  \label{AsPIIPole}
\end{equation}

In the focusing case, the reduction (\ref{FirstRed}) results in the modified
second Painlev\'{e} equation, 
\begin{equation}
F^{\prime \prime }=zF-2F^{3},  \label{PIIModified}
\end{equation}%
and any nontrivial real solution satisfies \cite{Clark10}:%
\begin{equation}
F\left( z\right) =r\left\vert z\right\vert ^{-1/4}\sin \left( s\left(
z\right) -\theta _{0}\right) +\text{O}\left( \left\vert z\right\vert
^{-5/4}\ln \left\vert z\right\vert \right) ,\quad z\rightarrow -\infty
\bigskip ,  \label{PIIAsyMod}
\end{equation}%
where%
\begin{equation}
s\left( z\right) =\frac{2}{3}\left\vert z\right\vert ^{3/2}+\frac{3}{4}%
r^{2}\ln \left\vert z\right\vert  \label{PIIModPhase}
\end{equation}%
with $r\neq 0$ and $\theta _{0}$ arbitrary real constants.

The second asymptotic is as follows. If%
\begin{equation}
\theta _{0}+\frac{3}{2}r^{2}\ln 2-\frac{1}{4}\pi -\arg \Gamma \left( \frac{i%
}{2}r^{2}\right) =\pi n,\qquad n=0,\pm 1,\pm 2,\ ...\ ,  \label{PIIModQuant}
\end{equation}%
we have%
\begin{equation}
F\left( z\right) \sim k_{0}\text{Ai\/}\left( z\right) ,\qquad z\rightarrow
+\infty \bigskip ,  \label{PIIAsModPlus}
\end{equation}%
where $k_{0}$ is a nonzero real constant. The connection formulas are 
\begin{equation}
r^{2}=\pi ^{-1}\ln \left( 1+k_{0}^{2}\right) ,\qquad \text{sign\/}\left(
k_{0}\right) =\left( -1\right) ^{n}.  \label{PIIModConnect}
\end{equation}

For the generic case, when the condition (\ref{PIIModQuant}) is not
satisfied, one gets \cite{Clark10}:%
\begin{equation}
F\left( z\right) =\alpha \sqrt{z/2}+\alpha \beta \left\vert 2z\right\vert
^{-1/4}\cos \left( s\left( z\right) +\theta \right) +\text{O}\left(
z^{-1}\right) ,\quad z\rightarrow +\infty \bigskip ,  \label{PIIAsDefGen}
\end{equation}%
where $\alpha ,$ $\beta >0$ and $\theta $ are real constants, and%
\begin{equation}
s\left( z\right) =\frac{1}{3}\left( 2z\right) ^{3/2}-\frac{3}{2}\beta
^{2}\ln z.  \label{PIIAsDefGenPh}
\end{equation}%
The connection formulas for $\alpha ,$ $\beta $ and $\theta $ are given by%
\begin{eqnarray}
&&\alpha =-\text{sign}\left( \func{Im}\xi \right) \text{\/},\qquad \beta
^{2}=\frac{1}{\pi }\ln \frac{1+\left\vert \xi \right\vert ^{2}}{2\left\vert 
\func{Im}\xi \right\vert },  \label{PIIConnectMod} \\
&&\theta =-\frac{3}{4}\pi -\frac{7}{2}\beta ^{2}\ln 2+\arg \left( 1+\xi
^{2}\right) +\arg \Gamma \left( i\beta ^{2}\right) ,  \notag
\end{eqnarray}%
where%
\begin{equation}
\xi =\left( \exp \left( \pi r^{2}\right) -1\right) ^{1/2}\exp \left( i\left( 
\frac{3}{2}r^{2}\ln 2-\frac{1}{4}\pi +\theta -\arg \Gamma \left( \frac{i}{2}%
r^{2}\right) \right) \right) .  \label{PIIConnectModKsi}
\end{equation}%
(See \cite{Clark10} and the references therein for more details.)

In sections 7.1--7.2, we discuss applications of the second Painlev\'{e}
transcendents to nonlinear waves in fiber optics, plasmas and ocean.

\section{Derivative Nonlinear Schr\"{o}dinger Equations}

The derivative nonlinear Schr\"{o}dinger equation describes the propagation
of circular polarized nonlinear Alfv\'{e}n waves in plasma physics \cite%
{Mioetal76}, \cite{Mioetal76a}, and \cite{Wadatietal78}. This equation is
completely integrable \cite{AbloRamSegII}, \cite{Ablo:Seg81}, \cite%
{ChenLeeLiu79}, \cite{KaupNewell78}, \cite{KakSasSat95}, \cite{KawIno78}, 
\cite{Lee84}, \cite{Lee89}, and \cite{NiaNing07} (see also the references
therein). Its different forms are related through the gauge transformations 
\cite{Kundu94}, \cite{WadatiSogo83}. Explicit solutions can be found, for
example, in Refs.~\cite{FengWang01}, \cite{KawIno78}, \cite{KawKobIno79}, 
\cite{KawSakKob80}, \cite{KengneLiu06}, \cite{Lashkin07}, and \cite{SunGao09}%
.

\section{Quintic Complex Ginzburg--Landau Equation}

A number of nonintegrable nonlinear dissipative PDEs are known to display a
wide variety of complex behavior, where the global time evolution is
governed by the dynamics of spatially localized structures. This was shown
in particular for the family of exact solutions of the one-dimensional
supercritical complex Ginzburg--Landau equation \cite{AransonKramer02}, \cite%
{BekkiNozaki94}, \cite{BekkiNozaki95}, \cite{Chate94}.

Exact solitary wave solutions of the one-dimensional quintic complex
Ginzburg--Landau equation are obtained in Refs.~\cite{MarcHugConte94} and 
\cite{SaaHoh92}. In their notation and terminology,%
\begin{equation}
\frac{\partial A}{\partial t}=\varepsilon A+\left( b_{1}+ic_{1}\right) \frac{%
\partial ^{2}A}{\partial x^{2}}-\left( b_{3}-ic_{3}\right) \left\vert
A\right\vert ^{2}A-\left( b_{5}-ic_{5}\right) \left\vert A\right\vert ^{4}A,
\label{QCG-LE}
\end{equation}%
where $\varepsilon ,$ $b_{1},$ $c_{1},$ $b_{3},$ $c_{3},$ $b_{3},$ $c_{3}$
are real constants and the field $A\left( x,t\right) $ is complex-valued
(see also \cite{Chodhury05}, \cite{Clark92}, \cite{Soto-Crestoetal97}, \cite%
{Solo-CrestoPes97}). These solutions are expressed in terms of hyperbolic
functions, and include coherent structures with a strong spatial
localization such as pulses and fronts, as well as, sources and sinks.
Equation (\ref{QCG-LE}) is a one-dimensional model of the large-scale
behavior of many nonequilibrium pattern-forming systems (see, for example, 
\cite{AransonKramer02}, \cite{BekkiNozaki95}, \cite{Chate94}, ~\cite%
{ChangLush11}, \cite{MarcHugConte94}, \cite{Saa03}, \cite{SaaHoh92} and the
references therein).

A systematic method for obtaining analytic solitary wave solutions of
nonintegrable PDEs has been introduced by Conte and Musette \cite%
{ConteMusette09}, \cite{MusetteConte03} and further developed by Hone \cite%
{Hone05a}, \cite{Hone05b} and Vernov \cite{Vernov06}, \cite{Vernov07} (see
also \cite{Clark92} for another approach and \cite{ConteNg12a}, \cite%
{ConteNg12b}, \cite{Gagnon89}, \cite{GagWint89}, \cite{KimMoon00} for exact
solutions). The unique elliptic traveling wave solution of \ (\ref{QCG-LE})
is found in Ref.~\cite{Vernov07}. Our approach allows to extend some of
these results to nonautonomous systems under consideration.

\section{Some Applications}

For applications of Schr\"{o}dinger equations in a variety of nonlinear
problems, see Refs.~\cite{Chiaoetal64}, \cite{ChenHH:LiuCS78}, \cite%
{ElGuretal93}, \cite{Gagnon89}, \cite{Gordon83}, \cite{GurevichNLIonosphere}%
, \cite{HasegawaTappert73a}, \cite{HasegawaTappert73b}, \cite{IchiImamTani72}%
, \cite{IchiTani73}, \cite{Kelley65}, \cite{KadomtsevCollectPP}, \cite%
{KadKarp71}, \cite{KarpmanNW}, \cite{Kovaleva08}, \cite{MikhOnishchTat85}, 
\cite{Mioetal76}, \cite{MuryShlyapetal02}, \cite{MoralesLee78}, \cite%
{Pokhotelovetal96}, \cite{Scott03}, \cite{Xuetal}, \cite{Wadatietal78}, \cite%
{Wadatietal79}, \cite{ZakhOstr09}, and \cite{Zakh:Shab71}.

\subsection{Accelerating Airy-Type Packets in Optics}

In a Kerr medium with cubic nonlinearity, for which the dependence of the
index of refraction on intensity is given by%
\begin{equation}
n=n_{0}\left( \omega \right) +i\chi \left( \omega \right) +n_{2}\left\vert
E\right\vert ^{2},  \label{nE2}
\end{equation}%
the propagation of optical pulses in single mode dispersive fibers is
described by the nonlinear Schr\"{o}dinger equation of the form \cite%
{Chiaoetal64}, \cite{HasegawaTappert73a}, \cite{HasegawaTappert73b}, \cite%
{SulemSulem99}:%
\begin{equation}
i\left( \frac{\partial \psi }{\partial t}+\omega _{0}^{\prime }\frac{%
\partial \psi }{\partial x}+\nu _{0}\psi \right) +\frac{1}{2}\omega
_{0}^{\prime \prime }\frac{\partial ^{2}\psi }{\partial x^{2}}+\upsilon 
\frac{\omega _{0}n_{2}}{n_{0}}\left\vert \psi \right\vert ^{2}\psi =0,
\label{NLSEFiber}
\end{equation}%
which is derived under assumptions that the complex envelope amplitude
function $\psi \left( x,t\right) $ varies slowly compared to the carrier and
that the nonlinear and dispersion terms are weak. Here,%
\begin{equation}
\omega _{0}^{\prime }=\frac{\partial \omega _{0}}{\partial k_{0}},\qquad
\omega _{0}^{\prime \prime }=\frac{\partial ^{2}\omega _{0}}{\partial
k_{0}^{2}},\qquad \nu _{0}=\chi \left( \omega _{0}\right) \frac{\omega _{0}}{%
n_{0}}  \label{NotationFiber}
\end{equation}%
and $\upsilon $ is a geometric factor which depends on the radial variation
of the guided electric field (we use the notation of Ref.~\cite%
{HasegawaTappert73a}). In equation (\ref{NLSEFiber}), the second term
describes the envelope propagation with group velocity $\omega _{0}^{\prime
} $ (in the absence of the rest of the terms), the third term describes the
effect of absorption, the fourth term describes the effect of dispersion,
and the fifth term describes the effect of nonlinearity (further details can
be found in \cite{HasegawaTappert73a}). This equation is usually called a
nonlinear parabolic equation in plasma physics \cite{GurevichNLIonosphere}, 
\cite{KadomtsevCollectPP}, and \cite{KadKarp71}.

In most of applications, equation\ (\ref{NLSEFiber}) has constant
coefficients and, therefore, is not integrable if $\nu _{0}\neq 0.$ A
similar special case of our generic equation (\ref{SchroedingerQuadratic}),
when $a=\omega _{0}^{\prime \prime }/2,$ $b=c=f=0,$ $g=-\omega _{0}^{\prime
},$ and $h=-\upsilon \omega _{0}n_{2}/n_{0}$ can be transformed into the
standard forms (\ref{DefFoc}) by Lemma~1 with $c_{0}=0.$ In our notations,
the standard solutions are given by $\mu _{0}=2ate^{2dt}$ and $\mu
_{1}=\left( 1-2dt\right) e^{2dt}.$ By (\ref{A0})--(\ref{F0}), $\lambda
=e^{2dt}$ and%
\begin{equation}
\alpha _{0}\left( t\right) =-2\beta _{0}\left( t\right) =\gamma _{0}\left(
t\right) =\frac{1}{4at},\qquad \delta _{0}=-\varepsilon _{0}=\frac{g}{2a}%
,\qquad \kappa _{0}\left( t\right) =\frac{g^{2}t}{4a}.  \label{A0F0Fiber}
\end{equation}%
Then $\mu =\mu \left( 0\right) \left( 1+4\alpha \left( 0\right) at\right)
e^{2dt}$ and%
\begin{equation}
\alpha =\frac{\alpha \left( 0\right) }{1+4\alpha \left( 0\right) at},\qquad
\beta =\frac{\beta \left( 0\right) }{1+4\alpha \left( 0\right) at},\qquad
\gamma =\gamma \left( 0\right) -\frac{\beta ^{2}\left( 0\right) at}{%
1+4\alpha \left( 0\right) at},  \label{InitData}
\end{equation}%
\begin{equation*}
\delta =\frac{\delta \left( 0\right) +2\alpha \left( 0\right) gt}{1+4\alpha
\left( 0\right) at},\quad \varepsilon =\varepsilon \left( 0\right) +\beta
\left( 0\right) \frac{g-2a\delta \left( 0\right) }{1+4\alpha \left( 0\right)
at}t,\quad \kappa =\kappa \left( 0\right) +\delta \left( 0\right) \frac{%
g-\delta \left( 0\right) }{1+4\alpha \left( 0\right) at}t
\end{equation*}%
by Lemma~2, thus providing the general transformation (\ref{PsiXT}) into the
standard forms (\ref{StNonLinSchrEq}). Finally, the integrability condition (%
\ref{SolInt2a}) takes the form%
\begin{equation}
h=h_{0}\beta ^{2}\left( 0\right) \mu \left( 0\right) \frac{ae^{2dt}}{%
1+4\alpha \left( 0\right) at}=a\left( 1+2t\left( d-2\alpha \left( 0\right)
a\right) t\right) +\text{O}\left( t^{2}\right) ,\quad t\rightarrow 0,
\end{equation}%
which indicates the most general integrable system of this kind and may be
achieved in fibers, say, through the geometrical factor $\upsilon $ in order
to support nonspreading signals. Moreover, our equations (\ref{InitData})
describe explicit evolution of variations of initial data. The choice of
initial condition $d=2\alpha \left( 0\right) a$ allows to include effects of
absorption into the zero approximation. Our modification of initial
condition allows to compensate some intensity losses.

The bright and dark soliton solutions are given in Refs.~ \cite%
{HasegawaTappert73a} and \cite{HasegawaTappert73b} (see also \cite%
{BushAnglin00}, \cite{MuryShlyapetal02}, \cite{Zakh:Shab71}--\cite%
{ZakhShab79} and the references therein). But here, we mainly concentrate on
solutions related to the second Painlev\'{e} transcendents.

Explicit solutions (\ref{FirstRed}) and (\ref{GereralRed}) in terms of the
nonlinear Airy functions reveal several remarkable features. In quantum
mechanical terms, the \textquotedblleft probability
density\textquotedblright\ $\left\vert \psi \left( x,t\right) \right\vert
^{2}$ not only remains unchanged in form but also continually accelerates in
empty space. Therefore, these solutions extend to nonlinear cases the
nonspreading Airy packet introduced by Berry and Balazs \cite{BerryBalazs79}
for the unidimensional linear Schr\"{o}dinger equation without potential
(see also \cite{Greenberg80} and \cite{UnnikishnanRau96}). There is no
violation of Ehrenfest's theorem because the Airy function is not square
integrable and cannot represent the probability density for a single
particle. What accelerates in the Airy packet is not any individual particle
but the caustic (i.e., the envelope, or focus) of the family of orbits (see 
\cite{BerryBalazs79} for more details on this geometrical approach). These
nonspreading and freely accelerating wave packets have recently been
demonstrated in both one- and two-dimensional configurations in optics \cite%
{SiviloglouChris07}, \cite{Siviloglouetal07} and we would like to elaborate
on this topic in detail.

It is worth noting that an expansion transformation from the Schr\"{o}dinger
group, e.g. formula (2.8) of \cite{Lop:Sus:VegaGroup} with $m=-1/t_{1},$
gives another accelerating solution of the free particle equation, $i\psi
_{t}+\psi _{xx}=0,$ in terms of Airy function as follows:%
\begin{eqnarray}
&&\psi \left( x,t\right) =\sqrt{\frac{\left\vert t_{1}\right\vert }{t_{1}-t}}%
\exp \left( ig\left( x-\frac{2g}{3}\frac{t_{1}t^{2}}{t_{1}-t}\right) \frac{%
t_{1}^{2}t}{\left( t_{1}-t\right) ^{2}}-\frac{ix^{2}}{4\left( t_{1}-t\right) 
}\right)  \label{AiryAccelerating} \\
&&\qquad \qquad \times g^{1/3}\text{Ai}\left( g^{1/3}\left( x-g\frac{%
t_{1}t^{2}}{t_{1}-t}\right) \frac{t_{1}}{t_{1}-t}\right) ,  \notag
\end{eqnarray}%
which is convenient for $t<t_{1}.$ The degenerate case, when $t_{1}=0,$ can
be analyzed with the help of transformation (2.9) of \cite{Lop:Sus:VegaGroup}%
:%
\begin{equation}
\psi \left( x,t\right) =\frac{1}{\sqrt{2t}}\exp \left( \frac{i}{4t}\left(
x^{2}+\left( x+\frac{g}{12t}\right) \frac{g}{2t}\right) \right) \ g^{1/3}%
\text{Ai}\left( -\frac{g^{1/3}}{2t}\left( x+\frac{g}{8t}\right) \right) .
\label{AiryAcceleratingZero}
\end{equation}%
From now on, we choose $t_{1}>0$ for the sake of simplicity. (The most
general solution of this kind can be obtained by transformation (2.5) of 
\cite{Lop:Sus:VegaGroup}.)

An arbitrary point Ai$\left( x_{0}=\text{constant}\right) $ on the Airy
function graph accelerates in empty space according to the law:%
\begin{eqnarray}
&&x\left( t\right) =\frac{gt_{1}^{3}}{t_{1}-t}-2gt_{1}^{2}+\left( gt_{1}+%
\frac{x_{0}}{g^{1/3}t_{1}}\right) \left( t_{1}-t\right) ,  \label{AiryGunSol}
\\
&&\ \ \frac{dx}{dt}=\frac{gt_{1}^{3}}{\left( t_{1}-t\right) ^{2}}-\left(
gt_{1}+\frac{x_{0}}{g^{1/3}t_{1}}\right) ,  \notag \\
&&\ \frac{d^{2}x}{dt^{2}}=\frac{2gt_{1}^{3}}{\left( t_{1}-t\right) ^{3}} 
\notag
\end{eqnarray}%
in the laboratory frame of reference. This \textquotedblleft
kinematics\textquotedblright\ can be obtained as a unique solution of the
following singular boundary value problem on $\left( -\infty ,t_{1}\right) $:%
\begin{eqnarray}
&&\qquad \left( t_{1}-t\right) ^{2}\frac{d^{2}x}{dt^{2}}-\left(
t_{1}-t\right) \frac{dx}{dt}-x=2gt_{1}^{2}  \label{NewtonLaw} \\
&&\left( \text{an analog of Newton's second law of motion}\right) ,  \notag
\end{eqnarray}%
\begin{eqnarray}
&&\qquad \qquad \qquad \lim_{t\rightarrow -\infty }\frac{x\left( t\right) }{%
t_{1}-t}=C_{1}=gt_{1}+\frac{x_{0}}{g^{1/3}t_{1}}  \label{NewtonVelocity} \\
&&\left( \frac{dx}{dt}\sim -C_{1},\text{ an asymptotic value of constant
velocity at }-\infty \right) ,  \notag
\end{eqnarray}%
\begin{eqnarray}
&&\qquad \qquad \lim_{t\rightarrow t_{1}^{-}}\left( t_{1}-t\right) x\left(
t\right) =C_{2}=gt_{1}^{3}\neq 0  \label{NewtonPole} \\
&&\left( x\sim C_{2}/\left( t_{1}-t\right) ,\text{ a given residue of the
simple pole at }t_{1}\right) .  \notag
\end{eqnarray}%
The corresponding family of orbits in phase space is given by the equation:%
\begin{equation}
\frac{Q^{2}}{C_{2}}=2C_{1}+P+\frac{C_{1}^{2}}{P},  \label{AiryOrbits}
\end{equation}%
where $P=dx/dt+C_{1}$ and $Q=x+2gt_{1}^{2}.$ By the quadratic formula,%
\begin{equation}
P=\frac{Q^{2}}{2C_{2}}-C_{1}\pm \sqrt{\frac{Q^{2}}{2C_{2}}\left( \frac{Q^{2}%
}{2C_{2}}-2C_{1}\right) }.  \label{AiryQuadratic}
\end{equation}%
This curve degenerates to a parabola when $C_{1}=0.$

In view of (\ref{AiryGunSol}), any two points on the graph, say Ai$\left(
x_{0}\right) $ and Ai$\left( y_{0}\right) $ with $x_{0}<y_{0}$ taken at the
same initial moment of time $t_{0},$ are contracting to each other:%
\begin{equation}
x\left( t\right) -y\left( t\right) =\left( x_{0}-y_{0}\right) \frac{t_{1}-t}{%
g^{1/3}t_{1}}\rightarrow 0\qquad \text{as\ }\ t\rightarrow t_{1}
\end{equation}%
and their trajectories have two different slopes as $t\rightarrow \pm \infty
,$ when dispersion that leads to decay occurs. In other words, these analogs
of classical trajectories\ have slanted asymptotes and linearly diverge from
each other as $t\rightarrow \pm \infty :$%
\begin{equation}
\lim_{t\rightarrow \pm \infty }\left( x\left( t\right) -\left( gt_{0}+\frac{%
x_{0}}{t_{1}}\right) \left( t_{1}-t\right) +2gt_{1}^{2}\right) =0,\qquad
\lim_{t\rightarrow t_{1}}\left( \left( t_{1}-t\right) x\left( t\right)
-gt_{1}^{3}\right) =0
\end{equation}%
but both of them approach to a common vertical asymptote at $t\rightarrow
t_{1}.$ With the aid of Airy function asymptotic \cite{Olver}, for any fixed
spatial point $x$ one gets%
\begin{equation}
\left\vert \psi \left( x,t\right) \right\vert \sim g^{1/3}\sqrt{\frac{2t_{1}%
}{\pi }}\left[ \left( t_{1}-t\right) ^{2}z\right] ^{-1/4}\left\vert \sin
\left( \frac{2}{3}z^{3/2}\right) \right\vert \qquad \text{as\ }t\rightarrow
t_{1}^{-}>0,  \label{AiryAsympt}
\end{equation}%
where%
\begin{equation}
z=g^{1/3}\left( \frac{gt_{1}t^{2}}{t_{1}-t}-x\right) \frac{t_{1}}{t_{1}-t}%
>0,\qquad \lim_{t\rightarrow t_{1}}\left( t_{1}-t\right)
^{2}z=g^{4/3}t_{1}^{4}.  \label{AiryAsymptArg}
\end{equation}%
Thus the leading term remains bounded and does not depend on the spatial
coordinate $x$ when $t\rightarrow t_{1}^{-}.$ As a result, a possible
\textquotedblleft blow up\textquotedblright\ of this solution cannot occur
at any finite point in space as expected in the linear case.

On the contrary, the cubic nonlinear Schr\"{o}dinger equation is no longer
preserved under the expansion transformation (2.8) of \cite%
{Lop:Sus:VegaGroup}. But the same symmetry, say by our Theorem~1, holds for
the quintic nonlinear Schr\"{o}dinger equation, which is thus invariant
under the action of this group of transformations. This is where the blow up
solutions do exist and some of them are analyzed in section 7.5 (see, for
example, equation (\ref{BlowUpPulces})).

Our solution (\ref{AiryAccelerating}) can, in principle, be generated from
the Airy packet of Berry and Balazs \cite{BerryBalazs79} in the following
fashion. Let us consider a generalized harmonic oscillator with the variable
Hamiltonian (\ref{SchroedingerQuadratic}), which evolves from the initial
free particle coefficients and terminates its evolution to the original case
of a free particle. Solution of the corresponding Riccati-type system is
provided by Lemma~1. By the end of the cycle, the original Airy packet may
possess the initial data of (\ref{AiryAccelerating}) and this solution will
be engaged by the continuity of the wave function. In addition, we provide
an example of instability for the linear Airy beam, namely, the
corresponding variation of the initial configuration:%
\begin{eqnarray}
&&\psi \left( x,t_{0}\right) =\sqrt{\frac{\left\vert t_{1}\right\vert }{%
t_{1}-t_{0}}}g^{1/3}\text{Ai}\left( g^{1/3}\left( x-\frac{gt_{1}t_{0}^{2}}{%
t_{1}-t_{0}}\right) \frac{t_{1}}{t_{1}-t_{0}}\right) \\
&&\qquad \qquad \times \exp \left( ig\left( x-\frac{2g}{3}\frac{%
t_{1}t_{0}^{2}}{t_{1}-t_{0}}\right) \frac{t_{1}^{2}t_{0}}{\left(
t_{1}-t_{0}\right) ^{2}}-\frac{ix^{2}}{4\left( t_{1}-t_{0}\right) }\right) ,
\notag
\end{eqnarray}%
results in a \textquotedblleft collapse\textquotedblright\ of the packet in
a finite time $t_{0}\leq t<t_{1}.$ In this process, any point from the
original beam attains an infinite velocity and acceleration as $t\rightarrow
t_{1}^{-}$ in the laboratory frame of reference due to deterministic
evolution of our solution (one may call this exotic dynamic an
\textquotedblleft Airy gun\textquotedblright ; a \textsl{Mathematica}
animation is available on the article website).

The study of propagating nonlinear Airy--Painlev\'{e} optical pulses in
dispersive fibers was initiated in \cite{GiaJos89} (see also Ref.~\cite%
{Smith76} for an earlier application of the second Painlev\'{e} transcendent
in hydrodynamics) and has been continued in recent publications \cite%
{BekSeg11}, \cite{KamSegChris11}, \cite{Lottietal11}, and \cite{RudMar11}.
Our solution (\ref{PsiXT}) incorporates the simplest nonhomogeneous effects
of a nonlinear time-varying media in a unified form. The corresponding
asymptotics, connection problems, and bibliography are given in section~4.3
for the reader's convenience. This summary may facilitate further use of
these results. An important example of motion of the nonlinear Airy packet
in a time-varying spatially uniform force is left to the reader (see \cite%
{BerryBalazs79} for discussion of the classical case).

It is worth addressing a few new features of the nonlinear Airy beams under
consideration. According to (\ref{AsPII}), in the defocusing case, which is
related to positive group velocity dispersion (dark pulse), a bounded on the
entire real line solution corresponds to the real parameter $\left\vert
k_{0}\right\vert <1,$ $k_{0}\neq 0$ (a finite range of the ratio of the
nonlinearity to dispersion \cite{GiaJos89}). In the focusing case (anomalous
dispersion, bright pulse), the bounded solution can only exist under
condition (\ref{PIIModQuant}), which may be thought of as a
\textquotedblleft quantization rule\textquotedblright\ in this nonlinear
problem. The connection relation take the form%
\begin{equation}
\theta _{n}=\frac{1}{4}\pi -\frac{3\ln 2}{2\pi }\ln \left(
1+k_{0}^{2}\right) +\arg \Gamma \left( \frac{i}{2\pi }\ln \left(
1+k_{0}^{2}\right) \right) +\pi n,\quad k_{0}\neq 0
\end{equation}%
and the asymptotics of bounded solutions of the modified second Painlev\'{e}
equation are given by%
\begin{equation}
F\left( z\right) =\left\{ 
\begin{array}{c}
\sim k_{0}\text{Ai\/}\left( z\right) ,\qquad z\rightarrow +\infty \bigskip ,
\\ 
\sqrt{\dfrac{\ln \left( 1+k_{0}^{2}\right) }{\pi }}\left\vert z\right\vert
^{-1/4}\sin \left( s\left( z\right) -\theta _{n}\right) +\text{O}\left(
\left\vert z\right\vert ^{-5/4}\ln \left\vert z\right\vert \right) ,\quad
z\rightarrow -\infty \bigskip%
\end{array}%
\right.
\end{equation}%
with%
\begin{equation}
s\left( z\right) =s\left( z\right) =\frac{2}{3}\left\vert z\right\vert
^{3/2}-\frac{3}{4\pi }\ln \left( 1+k_{0}^{2}\right) \ln \left\vert
z\right\vert
\end{equation}%
for $n=0,$ $\pm 1,$ $\pm 2,$ $\ldots \ .$ For the bright pulse, there is no
restriction on the nonlinear wave amplitude and this case deserves an
experimental study.

In our opinion, a joint consideration of (breaking/restoring) symmetry of
the linear and (cubic/quintic) nonlinear Schr\"{o}dinger equations provides
a natural hierarchy of accelerating, nonspreading, and blowing up wave
packets.

\subsection{Nonlinear Airy Waves in Plasma and Ocean}

In plasma physics, evolution of weakly nonlinear, strongly dispersive, short
wavelength electron Langmuir waves described by their modulating envelope is
governed by the Schr\"{o}dinger equation with cubic nonlinearity. The
corresponding solitons are called Langmuir solitons \cite{MoralesLee78}, 
\cite{Novikovetal} in the case of a homogeneous media. In laser plasma
experiments, the plasma is both inhomogeneous and nonlinear to the
electromagnetic waves and the large-amplitude plasma waves (see, for
example, \cite{ChenHH:LiuCS76} and the references therein for more details).
We demonstrate our approach in the simplest inhomogeneous situation.

Solitons moving with acceleration in linearly inhomogeneous plasma were
studied in Refs.~\cite{ChenHH:LiuCS76} and \cite{ChenHH:LiuCS78} (see also 
\cite{Balakrish85} and \cite{TappZab71}). Single solitons and multisolitons
are found accelerated whereas maintaining their shapes when moving around
even upon emerging from collisions with other solitons.

Here we analyze another interesting \textquotedblleft one
soliton\textquotedblright\ scenario that is somewhat complementary to
recently discovered Airy beams in optics. In the defocusing case,%
\begin{equation}
i\frac{\partial \psi }{\partial t}+\frac{\partial ^{2}\psi }{\partial x^{2}}%
-2kx\psi =2\left\vert \psi \right\vert ^{2}\psi ,  \label{LinForce}
\end{equation}%
the following simple substitution%
\begin{equation}
\psi \left( x,t\right) =e^{-2i\left( ktx+2k^{2}t^{3}/3\right) }\chi \left(
\xi ,\tau \right) ,\qquad \xi =x+2kt^{2},\quad \tau =t,
\end{equation}%
which is originally due to Tappert \cite{ChenHH:LiuCS76} (see also \cite%
{Suslov11} for a more general case), results in%
\begin{equation}
i\frac{\partial \chi }{\partial \tau }+\frac{\partial ^{2}\chi }{\partial
\xi ^{2}}=2\left\vert \chi \right\vert ^{2}\chi .
\end{equation}%
When\ $g=2k,$ the composition of this transformation with solution (\ref%
{NonlinearAiry}) predicts that, despite of the presence of the external
constant force, the following wave packet:%
\begin{equation}
\psi \left( x,t\right) =e^{i\left( x-vt/2\right)
v/2-igvt^{2}/2}g^{1/3}A_{k_{0}}\left( g^{1/3}\left( x-vt\right) \right)
\label{ConstVelosity}
\end{equation}%
moves with the constant speed $v$ (the standing wave solution, originally
found in Ref.~\cite{Smith76}, occurs if $v=0).$ Here, the soliton profiles
are defined as the Airy-type solutions of the second Painlev\'{e} equations
and the corresponding asymptotic properties of these nonlinear Airy
functions are discussed in section~4 for the both, defocusing and focusing,
cases (see also \cite{Clark10} for more details). It is worth noting once
again that in the focusing case the \textquotedblleft quantization
rule\textquotedblright\ (\ref{PIIModQuant}) should hold for the bounded
solution.

For an application of the second Painlev\'{e} transcendents, in the case of
the nonlinear Schr\"{o}dinger equation with the linear potential (\ref%
{LinForce}), to giant waves, as observed on the Agulhas Current of the
southwest Indian Ocean, see Ref.~\cite{Smith76}. The operation of an
electrostatic spherical probe in a slightly ionized, collision-dominated
plasma can also be described by the second Painlev\'{e} transcendents \cite%
{Cohen63}, \cite{Kash98}, and \cite{Kash04}. A similar behavior occurs in
numerical simulations of the shock waves in plasma and corresponding
experiments (see \cite{KadomtsevCollectPP}, \cite{KadKarp71}, \cite%
{KarpmanNW} and the references therein).

\subsection{Nonlinear Alfv\'{e}n waves}

The propagation of a nonlinear Alfv\'{e}n wave with a small but
not-vanishing wave number along the magnetic field in cold plasmas is
governed by the derivative nonlinear Schr\"{o}dinger equation \cite%
{Mioetal76}, \cite{Mioetal76a}, \cite{Wadatietal78}, \cite{Wadatietal79}.
The Alfv\'{e}n solitons are discussed in \cite{MjWyller86}, \cite%
{Pokhotelovetal96}, \cite{Xuetal}. The derivative nonlinear Schr\"{o}dinger
equation is also used for modeling of wave processes in nonlinear optics,
Stokes waves in fluids of finite depth, etc. Our approach give an
opportunity to incorporate the inhomogeneous effects of the media.

Freak waves in plasmas and optical rogue waves are discussed in recent
publications \cite{Arecchi11}, \cite{Lavederetal11}, \cite{Ruderman10}, \cite%
{ShuklaMoslem12}, and \cite{Xuetal12}.

\subsection{Ion Acoustic Envelope Solitons in Explosive Ionospheric
Experiments}

The theory of nonlinear wave modulation in cohesionless plasma on the basis
of the Vlasov description \cite{Vlasov38} with the nonlinear Landau dumping
is developed in \cite{IchiImamTani72} and \cite{IchiTani73} (see also \cite%
{WellandIchikawaWilh78}). The corresponding nonlinear Schr\"{o}dinger
equation that describes the amplitude of the plasma density fluctuations has
a nonlocal-nonlinear term \cite{IchiTani73}.

The stable ion acoustic envelop soliton propagating perpendicular to the
magnetic field lines can exist in the ionospheric plasma. These solitons
were identified by processing the data from the North Star active explosive
ionospheric experiment \cite{Kovaleva07}, \cite{Kovaleva08}. The parameters
of the soliton have been first estimated under assumption that the
coefficients in the Schr\"{o}dinger equation are constant \cite{Kovaleva07}.
The additional ion dissipating terms have been used in order to obtain the
stable soliton solutions in a plasma with steep density and temperature
gradients \cite{Kovaleva08} (see also the references therein).

Our approach give an alternative opportunity to incorporate nonautonomous
dissipative terms. Exact solutions of the cubic complex Ginzburg--Landau
equation are to be taken as a nonpeturbed approximation. These solutions
include coherent structures with a strong spatial localization such as
pulses and fronts, as well as, sources and sinks \cite{AransonKramer02}, 
\cite{ConteMusette93}, \cite{BekkiNozaki94}, \cite{BekkiNozaki95}, \cite%
{NozakiBekki83}, \cite{NozakiBekki84}, \cite{PerStein77}, and \cite{SaaHoh92}%
. More details will be discussed elsewhere.

\subsection{Oscillating Coherent Structures}

Let us consider the autonomous generalized nonlinear Schr\"{o}dinger
equation (\ref{DNLSEStandardForm}) and address an intriguing question ---
can self-oscillating solutions exist? By our Theorem~1, the following
special case $c_{0}=1,$ $d_{0}=d_{1}=d_{2}=0$ and $d_{3}=d_{4}$ of this
equation, namely,%
\begin{equation}
i\psi _{t}+\psi _{xx}-x^{2}\psi =id_{3}\left( \left\vert \psi \right\vert
^{2}\psi _{x}+\psi ^{2}\psi _{x}^{\ast }\right) +d_{5}\left\vert \psi
\right\vert ^{4}\psi ,  \label{HarmonicSol}
\end{equation}%
(we replace $\tau =-\gamma $ and use the original notation for convenience)
is invariant under the action of the Schr\"{o}dinger group, which was
originally introduced by Niederer \cite{Niederer73} for the linear harmonic
oscillator, when $d_{3}=d_{5}=0$ and space-oscillating solutions exist \cite%
{LopSusVegaHarm}, \cite{Marhic78}. Therefore the nonlinear equation may
possess some interesting oscillating solutions too! The direct action of the
Schr\"{o}dinger group is given by%
\begin{equation}
\psi \left( x,t\right) =\sqrt{\frac{\beta \left( 0\right) }{\left\vert
z\left( t\right) \right\vert }}\ e^{i\left( \alpha \left( t\right)
x^{2}+\delta \left( t\right) x+\kappa \left( t\right) \right) }\ \chi \left(
\beta \left( t\right) x+\varepsilon \left( t\right) ,-\gamma \left( t\right)
\right) ,  \label{SchroedingerGroup}
\end{equation}%
where, according to Lemma~5, we define $z\left( t\right)
=c_{1}e^{2it}+c_{2}e^{-2it}$ and find everything in terms of this
complex-valued function as follows%
\begin{eqnarray}
&&\alpha \left( t\right) =i\frac{\left( c_{1}c_{2}\right) ^{\ast
}z^{2}-c_{1}c_{2}\left( z^{\ast }\right) ^{2}}{2\left( c_{1}-c_{2}^{\ast
}\right) \left\vert z\right\vert ^{2}},\qquad \beta \left( t\right) =\pm 
\frac{\sqrt{\left\vert c_{1}\right\vert ^{2}-\left\vert c_{2}\right\vert ^{2}%
}}{2\left\vert z\right\vert ^{2}},\qquad \gamma \left( t\right) =\frac{1}{2}%
\arg z,  \label{SchroedingerGroupComplex} \\
&&\delta \left( t\right) =\frac{c_{3}z-c_{3}^{\ast }z^{\ast }}{2i\left\vert
z\right\vert ^{2}},\quad \varepsilon \left( t\right) =\pm \frac{%
c_{3}z+c_{3}^{\ast }z^{\ast }}{2\left\vert z\right\vert \sqrt{\left\vert
c_{1}\right\vert ^{2}-\left\vert c_{2}\right\vert ^{2}}},\quad \kappa \left(
t\right) =\frac{\left( c_{3}^{2}z+\left. c_{3}^{\ast }\right. ^{2}z^{\ast
}\right) \left( z-z^{\ast }\right) }{8i\left( c_{1}-c_{2}^{\ast }\right)
\left\vert z\right\vert ^{2}}.  \notag
\end{eqnarray}%
The complex parameters:%
\begin{equation}
c_{1}=\frac{1+\beta ^{2}\left( 0\right) }{2}-i\alpha \left( 0\right) ,\qquad
c_{2}=\frac{1-\beta ^{2}\left( 0\right) }{2}+i\alpha \left( 0\right) ,\qquad
c_{3}=\varepsilon \left( 0\right) \beta \left( 0\right) +i\delta \left(
0\right)  \label{ComplexParameters}
\end{equation}%
are defined in terms of `essential' real initial data (here we choose $%
\gamma \left( 0\right) =\kappa \left( 0\right) =0$ for the sake of
simplicity). The real form of transformation (\ref{SchroedingerGroup}) and
visualization of the corresponding oscillating \textquotedblleft
missing\textquotedblright\ solutions for the linear harmonic oscillator can
be found in Ref.~\cite{LopSusVegaHarm}. In view of%
\begin{equation}
\left\vert z\right\vert ^{2}=\left\vert c_{1}\right\vert ^{2}+\left\vert
c_{2}\right\vert ^{2}+c_{1}c_{2}^{\ast }e^{4it}+c_{1}^{\ast
}c_{2}e^{-4it},\quad c_{1}c_{2}^{\ast }=\frac{1-\beta ^{4}\left( 0\right) }{4%
}-\alpha ^{2}\left( 0\right) -i\alpha \left( 0\right) ,  \label{ZMod2}
\end{equation}%
the shape-preserving configurations, that are somewhat similar to the
coherent states of the linear harmonic oscillator, may occur through this
transformation only when $\alpha \left( 0\right) =0$ and $\beta \left(
0\right) =\pm 1.$

Our formulas (\ref{SchroedingerGroup})--(\ref{ComplexParameters}) provide a
new time-periodic solution of equation (\ref{HarmonicSol}) from any given
solution. Although explicit solutions of (\ref{HarmonicSol}) are not readily
available in the literature (see, for example, \cite{Conte07} and the
references therein for $d_{3}=0),$ Bose condensation and/or nonlinear
effects in \textquotedblleft non-Kerr materials\textquotedblright , e.g.
fiber optics beyond the cubic nonlinearity, provide important examples.

\subsubsection{Example~1}

Mean-field theory has been remarkable successful at description both static
and dynamic properties of Bose-Einstein condensates \cite%
{Dal:Giorg:Pitaevski:Str99}, \cite{PitString03Book}. The macroscopic wave
function obeys a $3D$ cubic nonlinear Schr\"{o}dinger equation, which is
usually called the Gross--Pitaevskii equation in this model. There are
several reasons to consider higher order nonlinearity in the
Gross--Pitaevskii equation \cite{PitString03Book}. The quintic case is of
particular importance because near Feshbach resonance one may turn the
scattering length to zero when the dominant interaction among atoms is due
to three-body effects (see \cite{BrazhKonotopPita06}, \cite{Kohler02}, \cite%
{MuryShlyapetal02}, \cite{Pit06}, \cite{Pollacketal10}, \cite%
{PollacketalHulet09} and the references therein; on $^{7}Li,$ for example,
the scattering length is reported as small as $0.01$ Bohr radii \cite%
{PollacketalHulet09}). Then the nonlinear term in the mean-field equation
has the quintic form. Another interesting example is a $1D$ Bose gas in the
limit of impenetrable particles \cite{Girardeau60}, \cite{Kolomeiskyetal00}, 
\cite{Tonks36}.

The $1D$-quintic nonlinear Schr\"{o}dinger equation without potential in
dimensionless units, 
\begin{equation}
iA_{t}+A_{xx}\pm \frac{3}{4}\left\vert A\right\vert ^{4}A=0,  \label{QNLSE}
\end{equation}%
has the following explicit solutions adapted from \cite{MarcHugConte94} (we
use the notation and terminology from \cite{MarcHugConte94} and \cite%
{SaaHoh92}; see also \cite{Gagnon89} and \cite{GagWint89}).

Pulses:%
\begin{eqnarray}
A\left( x,t\right) &=&e^{i\phi }\left[ \frac{k}{\cosh k\left( x-vt\right) }%
\right] ^{1/2}  \label{Pulse} \\
&&\times \exp i\left( \frac{vx}{2}+\left( k^{2}-v^{2}\right) \frac{t}{4}%
\right)  \notag
\end{eqnarray}%
($\phi ,$ $v$ and $k$ are real parameters, the upper sign of the nonlinear
term should be taken in (\ref{QNLSE}); see also \cite{Peletal96} and the
references therein). We have%
\begin{equation}
\int_{-\infty }^{\infty }\left\vert A\left( x,t\right) \right\vert ^{2}\
dx=\pi  \label{L2NormPulse}
\end{equation}%
and the corresponding plane wave expansion,%
\begin{equation}
A\left( x,t\right) =\frac{1}{\sqrt{2\pi }}\int_{-\infty }^{\infty
}e^{ipx}B\left( p,t\right) \ dp,  \label{PlanePulse}
\end{equation}%
can be found in terms of gamma functions:%
\begin{eqnarray}
&&B\left( p,t\right) =\frac{e^{i\phi }}{2\pi \sqrt{k}}\exp i\left( \frac{%
v^{2}+k^{2}}{4}-pv\right) t  \label{HyperPulse} \\
&&\qquad \qquad \times \ \Gamma \left( \frac{1}{4}+\frac{i}{2k}\left( p-%
\frac{v}{2}\right) \right) \Gamma \left( \frac{1}{4}-\frac{i}{2k}\left( p-%
\frac{v}{2}\right) \right)  \notag
\end{eqnarray}%
with the help of integral (\ref{BIntegral}) from Appendix (the case $\omega
=0$ allows to evaluate the $L^{1}$-norm of solution (\ref{Pulse})). The
energy functional is given by%
\begin{equation}
E=-\overline{\psi _{x}^{2}}-\frac{1}{4}\overline{\left\vert \psi \right\vert
^{4}}=\frac{v^{2}}{4}\geq 0
\end{equation}%
and its positivity provides a sufficient condition for developing of a blow
up, namely a singularity such that the wave amplitude tends to infinity in
finite time \cite{ChangLush11}, \cite{SulemSulem99}, and \cite{ZakhSyn76}.

Sources and sinks:%
\begin{eqnarray}
A\left( x,t\right) &=&e^{i\phi }r^{1/2}\left[ \frac{\cosh \left( \sqrt{3}%
r\left( x-vt\right) \right) \pm 1}{\cosh \left( \sqrt{3}r\left( x-vt\right)
\right) \mp 2}\right] ^{1/2}  \label{Source} \\
&&\times \exp i\left( \frac{vx}{2}-\left( v^{2}+3r^{2}\right) \frac{t}{4}%
\right)  \notag
\end{eqnarray}%
($\phi ,$ $v$ and $r$ are real parameters; see also \cite{Kolomeiskyetal00}%
). Equation (\ref{QNLSE}) has also a class of (double) periodic solutions in
terms of elliptic functions \cite{Gagnon89}, \cite{GagWint89}. They will be
discussed elsewhere.

In addition, direct action of the Schr\"{o}dinger group \cite%
{Lop:Sus:VegaGroup}, \cite{Niederer72} on (\ref{Pulse}) produces a
six-parameter family of square integrable solutions:%
\begin{eqnarray}
\psi \left( x,t\right) &=&\sqrt{\frac{\beta \left( 0\right) }{1+4\alpha
\left( 0\right) t}}\ \exp i\left( \frac{\alpha \left( 0\right) x^{2}+\delta
\left( 0\right) x-\delta ^{2}\left( 0\right) t}{1+4\alpha \left( 0\right) t}%
+\kappa \left( 0\right) \right)  \notag \\
&&\times A\left( \beta \left( 0\right) \frac{x-2\delta \left( 0\right) t}{%
1+4\alpha \left( 0\right) t}+\varepsilon \left( 0\right) ,\ \frac{\beta
^{2}\left( 0\right) t}{1+4\alpha \left( 0\right) t}-\gamma \left( 0\right)
\right)  \label{BlowUpPulces}
\end{eqnarray}%
(one can choose $v=0$ withot loss of generality; see also \cite%
{Ald:Coss:Guerr:Lop-Ru11}, \cite{Ghosh02}, \cite{Guerr:Lop:Ald:Coss11} and
the references therein regarding the Schr\"{o}dinger group). All of them
blow up at the point $x_{0}=-\delta \left( 0\right) /2\alpha \left( 0\right) 
$ in a finite time, when $t\rightarrow t_{0}=-1/4\alpha \left( 0\right) $
and $\alpha \left( 0\right) \neq 0.$ (We use real-valued initial data for
the corresponding Riccati-type system; see \cite{Lop:Sus:VegaGroup} for more
details.) Blow up of these solutions can be naturally studied in phase
space. The corresponding Wigner function \cite{HilletyetalWigner84}, \cite%
{Wigner32} is evaluatited with the help of integral (\ref{IntegralFourier}).

A finite time blow up of solutions of the unidimensional quintic nonlinear
Schr\"{o}dinger equation (\ref{QNLSE}) is a classical result discussed in
many publications; see, for example, \cite{BudZakhSyn75}, \cite{ChangLush11}%
, \cite{Peletal96}, \cite{Rubinetal00}, \cite{ZakhSobSyn71}, \cite{ZakhSyn76}
and the references therein. This case is critical because any decrease of
the power of nonlinearity results in the global existence of solutions \cite%
{GinVelo79}, \cite{SulemSulem99}, \cite{Weinstein83}, \cite{ZakhSyn76} (see
also Refs. \cite{KolomStraley92} and \cite{Kolomeiskyetal00} for a
renormalization approach). Here, for the family of solutions (\ref%
{BlowUpPulces}), we have shown that the blow up occurs due to a
\textquotedblleft hidden symmetry\textquotedblright\ of this nonlinear PDE.
This property holds for all solutions which exponentially decay at infinity.

\subsubsection{Example~2}

The quintic nonlinear Schr\"{o}dinger equation in a parabolic confinement,%
\begin{equation}
i\psi _{t}+\psi _{xx}-x^{2}\psi \pm \frac{3}{4}\left\vert \psi \right\vert
^{4}\psi =0,  \label{QNLSEParabolic}
\end{equation}%
describes a mean-field model of strongly interacting $1D$ Bose gases\ for
the practically important case of a harmonic trap \cite{BrazhKonotopPita06}, 
\cite{DunLorOlsh01}, \cite{GirardeauWrite00}, \cite{Kolomeiskyetal00}, \cite%
{KolomStraley92}, \cite{Pit06}, \cite{Rubinetal00} and, in particular the
so-called \textit{Tonks--Girardeau} gas of impenetrable bosons \cite%
{Girardeau60}, \cite{Tonks36}; see \cite{KinoWegWeiss04} and \cite%
{ParedesetalShlyap04} for experimental observations. (The time-independent
version of the quintic nonlinear Schr\"{o}dinger equation has been
rigorously derived from the many-body problem \cite{LiebSeiringYng03}; see
also \cite{LiebSeiringYng00} for a rigorous derivation of the
Gross-Pitaevskii energy functional.)

Our observation is as follows. By the gauge transformation (e.~g. \cite%
{Ald:Coss:Guerr:Lop-Ru11}, \cite{Guerr:Lop:Ald:Coss11}, \cite%
{Lop:Sus:VegaGroup} and the references therein for the linear problem, the
quintic nonlinearity is invariant under this transformation by our
Theorem~1), equation (\ref{QNLSEParabolic}) has the following solution:%
\begin{equation}
\psi \left( x,t\right) =\frac{e^{-\left( i/2\right) x^{2}\tan 2t}}{\sqrt{%
\cos 2t}}A\left( \frac{x}{\cos 2t},\frac{\tan 2t}{2}\right) ,
\label{GaugeFreeOscillator}
\end{equation}%
where $A\left( x,t\right) $ is any solution of (\ref{QNLSE}), in particular,
the pulses and sources (\ref{Pulse}) and (\ref{Source}).

Oscillating pulses:%
\begin{eqnarray}
\psi \left( x,t\right) &=&e^{i\phi }\sqrt{\frac{2k}{\cos 2t}}\ \text{sech}%
^{1/2}\left( \frac{2k}{\cos 2t}\left( x-v\sin 2t\right) \right)
\label{OscPulse} \\
&&\times \exp i\frac{2vx+\left( k^{2}-v^{2}-x^{2}\right) \sin 2t}{2\cos 2t} 
\notag
\end{eqnarray}%
($\phi ,$ $v$ and $k$ are real parameters, the upper sign should be taken in
the nonlinear term). They are square integrable at all times:%
\begin{equation}
\frac{1}{\pi }\int_{-\infty }^{\infty }\left\vert \psi \left( x,t\right)
\right\vert ^{2}\ dx=1  \label{PI}
\end{equation}%
and%
\begin{eqnarray}
\frac{1}{\pi }\int_{-\infty }^{\infty }\left\vert x\psi \left( x,t\right)
\right\vert ^{2}\ dx &=&\frac{\pi ^{2}}{\left( 4k\right) ^{2}}\cos
^{2}2t+v^{2}\sin ^{2}2t,  \label{XII} \\
\frac{1}{\pi }\int_{-\infty }^{\infty }\left\vert \psi _{x}\left( x,t\right)
\right\vert ^{2}\ dx &=&\frac{\pi ^{2}}{\left( 4k\right) ^{2}}\sin
^{2}2t+v^{2}\cos ^{2}2t+\frac{k^{2}}{2\cos ^{2}2t}.  \label{PSII}
\end{eqnarray}%
The corresponding elementary integrals are collected in Appendix for
convenience; see (\ref{SechElemInts}).

The expectation values and variances of the position $x$ and momentum $%
p=i^{-1}\partial /\partial x$ operators are given by%
\begin{equation}
\overline{x}=\frac{\left\langle x\right\rangle }{\left\langle 1\right\rangle 
}=v\sin 2t,\qquad \overline{p}=\frac{\left\langle p\right\rangle }{%
\left\langle 1\right\rangle }=v\cos 2t  \label{ExpectXP}
\end{equation}%
and%
\begin{equation}
\left( \delta x\right) ^{2}=\overline{x^{2}}-\left. \overline{x}\right. ^{2}=%
\frac{\pi ^{2}}{\left( 4k\right) ^{2}}\cos ^{2}2t,\qquad \left( \delta
p\right) ^{2}=\overline{p^{2}}-\left. \overline{p}\right. ^{2}=\frac{\pi ^{2}%
}{\left( 4k\right) ^{2}}\sin ^{2}2t+\frac{k^{2}}{2\cos ^{2}2t},
\label{VarXP}
\end{equation}%
respectively. [It is worth noting that%
\begin{equation}
\overline{x}=vt,\qquad \overline{p}=\frac{v}{2},\qquad \left( \delta
x\right) ^{2}=\frac{\pi ^{2}}{4k^{2}},\qquad \left( \delta p\right) ^{2}=%
\frac{k^{2}}{8}  \label{XPpulse}
\end{equation}%
with%
\begin{equation}
\left( \delta p\right) ^{2}\left( \delta x\right) ^{2}=\frac{\pi ^{2}}{32}>%
\frac{1}{4}  \label{HeisenbergPulse}
\end{equation}%
for the original traveling wave solution (\ref{Pulse})]. The energy
functional is given by%
\begin{equation}
E=\overline{H}=\overline{p^{2}}+\overline{x^{2}}-\frac{1}{4}\overline{%
\left\vert \psi \right\vert ^{4}}=\frac{\pi ^{2}}{\left( 4k\right) ^{2}}%
+v^{2}>0  \label{EnergyPulseOsc}
\end{equation}%
by the direct evaluation.

A remarkable feature of the oscillating solution (\ref{OscPulse}) is that
the corresponding probability density converges, say as a sequence,
periodically in time, to the Dirac delta function at the turning points: $%
\left\vert \psi \left( x,t\right) \right\vert ^{2}\rightarrow \pi \delta
\left( x\mp v\right) $ as $t\rightarrow \pm \pi /4$ etc., when an
\textquotedblleft absolute squeezing\textquotedblright\ and/or localization,
namely $\delta x=0,$ occurs with $\delta p=\infty .$ The fundamental
Heisenberg uncertainty principle holds%
\begin{equation}
\left( \delta p\right) ^{2}\left( \delta x\right) ^{2}=\frac{\pi ^{2}}{32}%
\left( 1+\frac{\pi ^{2}}{32k^{4}}\sin ^{2}4t\right) \geq \frac{\pi ^{2}}{32}>%
\frac{1}{4}  \label{HeisenbergUP}
\end{equation}%
at all times. (It is worth noting that $\pi ^{2}/8\approx \allowbreak
1.\allowbreak 2337.$ The minimum-uncertainty squeezed states for a linear
harmonic oscillator, when the absolute minimum of the product can be
achieved, are discussed in Ref.~\cite{KrySusVegaMinimum}.)

The corresponding Wigner function \cite{HilletyetalWigner84}, \cite{Wigner32}%
:%
\begin{equation}
W\left( x,p,t\right) =\frac{1}{2\pi }\int_{-\infty }^{\infty }\psi ^{\ast
}\left( x+y/2,t\right) \psi \left( x-y/2,t\right) e^{ipy}\ dy,
\label{WignerDefinition}
\end{equation}%
can be evaluated in terms of hypergeometric function:%
\begin{equation}
W\left( x,p,t\right) =\text{sech\ }\omega \ _{2}F_{1}\left( 
\begin{array}{c}
1/2+i\omega ,\ 1/2-i\omega \medskip \\ 
1%
\end{array}%
;-\sinh ^{2}\vartheta \right)  \label{Wigner}
\end{equation}%
with the aid of integral (\ref{IntegralFourier}). Here%
\begin{equation}
\omega =\frac{1}{2k}\left( p\cos 2t+x\sin 2t-v\right) ,\qquad \vartheta =%
\frac{2k}{\cos 2t}\left( x-v\sin 2t\right) .  \label{WignerParameters}
\end{equation}%
A visualization of Wigner's function in phase space can be found at the
article website.

Our mathematical example reveals a surprising result that a medium described
by the quintic nonlinear Schr\"{o}dinger equation (\ref{QNLSEParabolic}) may
allow, in principle, to measure the coordinate of a \textquotedblleft
particle\textquotedblright\ with any accuracy, below the so-called vacuum
noise level and without violation of the Heisenberg uncertainty relation,
which is a major obstacle, for example, in the direct detection of
gravitational waves \cite{Abadetal11}, \cite{Eberleetal10}.

Oscillating sources and sinks:%
\begin{eqnarray}
\psi \left( x,t\right) &=&e^{i\phi }\sqrt{\frac{2r}{3^{1/2}\cos 2t}}\left[ 1-%
\frac{3}{\cosh \left( \dfrac{2r}{\cos 2t}\left( x-v\sin 2t\right) \right) +2}%
\right] ^{1/2}  \label{OscSource} \\
&&\times \exp i\frac{2vx-\left( v^{2}+r^{2}+x^{2}\right) \sin 2t}{2\cos 2t} 
\notag
\end{eqnarray}%
($\phi ,$ $v$ and $r$ are real parameters, we have chosen the lower sign of
the nonlinear term in (\ref{QNLSEParabolic})). Their detailed investigation
will be given elsewhere.

Then the action of Schr\"{o}dinger group, say in our complex form (\ref%
{SchroedingerGroup})--(\ref{ComplexParameters}), on (\ref{OscPulse}) and/or (%
\ref{OscSource}) produces a six-parameter family of new oscillating
solutions of equation (\ref{QNLSEParabolic}). For example, the following
extension of (\ref{OscPulse}) holds:%
\begin{eqnarray}
\psi \left( x,t\right) &=&e^{i\phi }\sqrt{\frac{2k\beta \left( 0\right) }{%
2\alpha \left( 0\right) \sin 2t+\cos 2t}}  \label{OscPulseGeneral} \\
&&\times \ \text{sech}^{1/2}2k\left( \beta \left( 0\right) \frac{x-\left(
\delta \left( 0\right) +v\beta \left( 0\right) \right) \sin 2t}{2\alpha
\left( 0\right) \sin 2t+\cos 2t}+\varepsilon \left( 0\right) \right)  \notag
\\
&&\times \exp i\frac{\left( 2\alpha \left( 0\right) \cos 2t-\sin 2t\right)
x^{2}+\delta \left( 0\right) \left( 2x-\delta \left( 0\right) \sin 2t\right) 
}{2\left( 2\alpha \left( 0\right) \sin 2t+\cos 2t\right) }  \notag \\
&&\times \exp i\left[ \beta \left( 0\right) \frac{2v\left( x-\delta \left(
0\right) \sin 2t\right) +\left( k^{2}-v^{2}\right) \beta \left( 0\right)
\sin 2t}{2\left( 2\alpha \left( 0\right) \sin 2t+\cos 2t\right) }%
+v\varepsilon \left( 0\right) \right] ,  \notag
\end{eqnarray}%
which presents the most general solution of this kind. (We assume that $%
\gamma \left( 0\right) =\kappa \left( 0\right) =0$ for the sake of
simplicity. Although the breather/pulsing solution, when $\alpha \left(
0\right) =\delta \left( 0\right) =\varepsilon \left( 0\right) =v=0$ and $%
\beta \left( 0\right) =1,$ was found in Ref.~\cite{Rubinetal00}, our
discussion of the uncertainty relation and Wigner function seems to be
missing in the available literature.) The blows up occur periodically in
time at the points 
\begin{equation}
x_{0}=\pm \frac{\delta \left( 0\right) +v\beta \left( 0\right) }{\sqrt{%
4\alpha ^{2}\left( 0\right) +1}},\qquad \text{when}\quad \cot 2t=-2\alpha
\left( 0\right) .
\end{equation}%
The corresponding Wigner function is given by our formula (\ref{Wigner})
with the following values of parameters:%
\begin{eqnarray}
&&\omega =\frac{1}{2k\beta \left( 0\right) }\left[ \left( p-2\alpha \left(
0\right) x\right) \cos 2t+\left( 2\alpha \left( 0\right) p+x\right) \sin
2t-\delta \left( 0\right) -v\beta \left( 0\right) \right] ,
\label{WignerParametersGeneral} \\
&&\vartheta =2k\left( \beta \left( 0\right) \frac{x-\left( \delta \left(
0\right) +v\beta \left( 0\right) \right) \sin 2t}{2\alpha \left( 0\right)
\sin 2t+\cos 2t}+\varepsilon \left( 0\right) \right) .  \notag
\end{eqnarray}%
(One can put $v=0$ in (\ref{OscPulseGeneral})--(\ref{WignerParametersGeneral}%
) without loss of generality because the general action of the Schr\"{o}%
dinger group already includes the Galilei transformation \cite{Niederer72}, 
\cite{Niederer73}.)

\subsection{A Nonlinear Harmonic Oscillator and Painlev\'{e} IV}

The following special case of equation (\ref{DNLSEStandardForm}):%
\begin{equation}
i\psi _{t}+\psi _{xx}-\frac{1}{4}x^{2}\psi =2x\left\vert \psi \right\vert
^{2}\psi +3\left\vert \psi \right\vert ^{4}\psi ,  \label{NLSEHarmOscillator}
\end{equation}%
by the substitution%
\begin{equation}
\psi \left( x,t\right) =e^{-i\left( n+1/2\right) t}u\left( x\right)
\end{equation}%
can be reduced to the ordinary differential equation for a nonlinear
harmonic oscillator studied in Refs.~\cite{Bassometal92}, \cite{Bassometal93}%
, and \cite{Clark10}:%
\begin{equation}
u^{\prime \prime }=3u^{5}+2xu^{3}+\left( \frac{1}{4}x^{2}-n-\frac{1}{2}%
\right) u.
\end{equation}%
There are exact `bounded state' solutions for any integer $n=0,1,2,\dots \ $
that asymptotically approach the wave functions of the linear harmonic
oscillator as $\left\vert x\right\vert \rightarrow \infty $ (see \cite%
{Clark10} for a summary of these results). The probability density\ $%
\left\vert \psi \left( x,t\right) \right\vert ^{2}=u^{2}\left( x\right) $
can be found in terms of a special fourth Painlev\'{e} transcendent.
Equation (\ref{NLSEHarmOscillator}) arises also as a symmetry reduction of
the derivative nonlinear Schr\"{o}dinger equation, which is solvable by
inverse scattering techniques \cite{AbloRamSegII}, \cite{Bassometal92}, and 
\cite{KaupNewell78}. (A symmetry reduction of the cubic-quintic nonlinear
Schr\"{o}dinger equation to fourth Painlev\'{e} transcendents is discussed
in \cite{GagWint89}.) There are a very few nonlinear ordinary differential
equations which obey the Painlev\'{e} property and have bounded solutions
and this is one of those. It would be interesting to find a connection of
these solutions investigated in Refs.~\cite{Bassometal92}--\cite%
{Bassometal93} in details with the theory of Bose--Einstein condensation and
fiber optics, say for a stable signal propagation.

\noindent \textbf{Acknowledgments.\/} One of the authors (SKS) thanks
Michael Berry, Carlos Castillo-Ch\'{a}vez, Robert Conte, Tom Koornwinder,
Oleksandr Pavlyk, Peter Paule, Georgy V.~Slyapnikov, Erwin Suazo, and Jos%
\'{e} M. Vega-Guzm\'{a}n for valuable discussions. He is also grateful to
Research Institute for Symbolic Computation, Johannes Kepler Universit\"{a}t
Linz, Austria, where some part of this work was done, for their hospitality.

\appendix

\section{Integral Evaluations}

The following integral%
\begin{equation}
\int_{-\infty }^{\infty }\frac{e^{i\omega s}\ ds}{\sqrt{\cosh s+\cosh c}}=%
\frac{\sqrt{2}\pi }{\cosh \pi \omega }\ _{2}F_{1}\left( 
\begin{array}{c}
\dfrac{1}{2}+i\omega ,\ \dfrac{1}{2}-i\omega \medskip \\ 
1%
\end{array}%
;\ -\sinh ^{2}\frac{c}{2}\right) ,\quad \left\vert \sinh \frac{c}{2}%
\right\vert <1,  \label{IntegralFourier}
\end{equation}%
which \textsl{Mathematica} fails to evaluate in a compact form\footnote{%
Oleksandr Pavlyk gave a \textsl{Mathematica} evaluation of this integral.},
can be derived as a special case of integral representation (2) on page 82
of Ref.~\cite{Erd}. The hypergeometric function is related to the Legendre
associated functions, which are a special case of Jacobi functions, see \cite%
{Dunster}, \cite{Koornwinder75}, \cite{Koornwinder85}:%
\begin{equation}
P_{1/2-i\omega }\left( \cosh c\right) =\ _{2}F_{1}\left( 
\begin{array}{c}
\dfrac{1}{2}+i\omega ,\ \dfrac{1}{2}-i\omega \medskip \\ 
1%
\end{array}%
;\ -\sinh ^{2}\frac{c}{2}\right)  \label{Legendre}
\end{equation}%
and Mehler conical functions. An important special case is given by%
\begin{equation}
\int_{-\infty }^{\infty }\frac{e^{i\omega s}}{\sqrt{\cosh s}}\ ds=\frac{\
\Gamma \left( 1/4+i\omega /2\right) \Gamma \left( 1/4-i\omega /2\right) }{%
\sqrt{2\pi }}.  \label{BIntegral}
\end{equation}%
Useful elementary integrals,%
\begin{equation}
\int_{-\infty }^{\infty }\frac{du}{\cosh u}=\pi ,\quad \int_{-\infty
}^{\infty }\frac{u^{2}\ du}{\cosh u}=\frac{\pi ^{3}}{4},\quad \int_{-\infty
}^{\infty }\frac{\sinh ^{2}u\ du}{\cosh ^{3}u}=\frac{3\pi }{8},\quad
\int_{-\infty }^{\infty }\frac{du}{\cosh ^{3}u}=\frac{\pi }{2},
\label{SechElemInts}
\end{equation}%
can be verified by \textsl{Mathematica}.

\end{document}